\documentclass[11pt,letterpaper]{article}
\usepackage{thmtools}
\usepackage{mathtools}
\usepackage{thm-restate}
\usepackage{bm}
\usepackage{mathrsfs}           %
\usepackage{vmargin,fancyhdr}   %
\usepackage{tikz}
\usetikzlibrary{positioning,chains,fit,shapes,calc}

\usepackage{subcaption}
\usepackage{algorithm}
\usepackage{algorithmic}
\usepackage{enumerate}

\usepackage{amsmath,amssymb, amsthm}    %
\usepackage{verbatim}           %
\usepackage{xspace}             %
\usepackage{graphicx,float}     %
\usepackage{ifthen,calc}        %
\usepackage{textcomp}           %
\usepackage{fancybox}           %
\usepackage{hhline}             %
\usepackage{float}              %

\usepackage[vflt]{floatflt}     %
\usepackage[small,compact]{titlesec}
\usepackage{setspace}
\usepackage{color}
\definecolor{ForestGreen}{rgb}{0.1333,0.5451,0.1333}
\usepackage{thumbpdf}
\usepackage[letterpaper,
colorlinks,linkcolor=ForestGreen,citecolor=ForestGreen,
backref,
bookmarks,bookmarksopen,bookmarksnumbered]
{hyperref}

\setpapersize{USletter}
\setmarginsrb{1in}{.5in}        %
{1in}{.5in}        %
{.25in}{.25in}     %
{.25in}{.5in}      %
\newcommand{\showccc}[0]{0}
\newcommand{\ccc}[2][nothing]{%
	\ifthenelse{\showccc=0}{}{
		\ensuremath{^{\Lsh\Rsh}}\marginpar{\raggedright\tiny\textsf{%
				\ifthenelse{\equal{#1}{nothing}}{}{\textbf{#1}\\}#2}}}}
\newcounter{hours}\newcounter{minutes}
\newcommand{\hhmm}{%
	\setcounter{hours}{\time/60}%
	\setcounter{minutes}{\time-\value{hours}*60}%
	\ifthenelse{\value{hours}<10}{0}{}\thehours:%
	\ifthenelse{\value{minutes}<10}{0}{}\theminutes}
\ifthenelse{\showccc=0}{\rhead{}}{\rhead{\today \ [\hhmm]}}
\newtheorem{theorem}{Theorem}

\newtheorem{proposition}{Proposition}
\newtheorem{corollary}{Corollary}

\newtheorem{lemma}{Lemma}
\newtheorem{fact}{Fact}

\newtheorem{problem}{Problem}
\newtheorem{assumption}{Assumption}

\usepackage{float}
\floatstyle{plain}\newfloat{myfig}{t}{figs}[section]
\floatname{myfig}{\textsc{Figure}}
\floatstyle{plain}\newfloat{myalg}{H}{algs}[section]
\floatname{myalg}{}
\newcommand{\defeq}{:=}
\newcommand{\norm}[1]{\left\lVert#1\right\rVert}
\newcommand{\inprod}[2]{\left\langle#1, #2\right\rangle}
\newcommand{\eps}{\epsilon}
\newcommand{\lam}{\lambda}
\newcommand{\argmax}{\textup{argmax}}
 
\newcommand{\R}{\mathbb{R}}

\newcommand{\half}{\frac{1}{2}}
\newcommand{\thalf}{\tfrac{1}{2}}
\newcommand{\1}{\mathbf{1}}
\newcommand{\E}{\mathbb{E}}

\newcommand{\Nor}{\mathcal{N}}

\newcommand{\Tr}{\textup{Tr}}
\newcommand{\opt}{\textup{OPT}}

\newcommand{\ma}{\mathbf{A}}
\newcommand{\mb}{\mathbf{B}}
\newcommand{\ms}{\mathbf{S}}
\newcommand{\my}{\mathbf{Y}}

\newcommand{\id}{\mathbf{I}}
\newcommand{\tO}{\widetilde{O}}
\newcommand{\nnz}{\textup{nnz}}
\newcommand{\lmax}{\lambda_{\textup{max}}}
\newcommand{\lmin}{\lambda_{\textup{min}}}

\newcommand{\Par}[1]{\left(#1\right)}
\newcommand{\Brack}[1]{\left[#1\right]}
\newcommand{\Brace}[1]{\left\{#1\right\}}
\newcommand{\covar}{\boldsymbol{\Sigma}}
\newcommand{\mzero}{\mathbf{0}}

\newcommand{\Sym}{\mathbb{S}}
\newcommand{\PSD}{\Sym_{\ge 0}}
\newcommand{\dist}{\mathcal{D}}
\newcommand{\fS}{\mathfrak{S}}
\newcommand{\veps}{\varepsilon}
\newcommand{\mprox}{\boldsymbol{\Gamma}}
\newcommand{\tX}{\widetilde{X}}
\newcommand{\net}{\mathcal{N}}
\newcommand{\teps}{\tilde{\eps}}
\newcommand{\mm}{\mathbf{M}}
\newcommand{\mn}{\mathbf{N}}
\newcommand{\mmg}{\mathbf{M}_G}
\newcommand{\mmb}{\mathbf{M}_B}
\newcommand{\wg}{w_G}
\newcommand{\wb}{w_B}
\newcommand{\PackingLP}{\mathsf{PackingLP}}
\newcommand{\lppacking}{\mathsf{PNormPacking}}
\newcommand{\schattenpacking}{\mathsf{SchattenPacking}}
\newcommand{\mv}{\mathbf{V}}
\newcommand{\mq}{\mathbf{Q}}
\newcommand{\mz}{\mathbf{Z}}
\newcommand{\one}{\mathbf{1}}
\newcommand{\tw}{\widetilde{w}}
\newcommand{\alla}{\mathcal{A}}
\newcommand{\nomega}{\Omega\Par{\tfrac{d + \log \delta^{-1}}{(\eps \log \eps^{-1})^2}}}
\newcommand{\RobustPCA}{\mathsf{RobustPCA}}
\newcommand{\BoxPacking}{\mathsf{BoxedSchattenPacking}}
\newcommand{\PCAFilter}
	{\mathsf{PCAFilter}}
	
\newcommand{\OneD}{\mathsf{1DRobustVariance}}

\definecolor{burntorange}{rgb}{0.8, 0.33, 0.0}

\newcommand{\jerry}[1]{\textcolor{red}{jerry: #1}}
  \usepackage{nth}
  \usepackage{intcalc}

\usepackage{url}
\usepackage{enumitem}

\begin{document}

	\begin{titlepage}
		\def\thepage{}
		\thispagestyle{empty}
		
		\title{Robust Sub-Gaussian Principal Component Analysis\\ and Width-Independent Schatten Packing} 
		
		\date{}
		\author{
			Arun Jambulapati\thanks{Stanford University, {\tt \{jmblpati, kjtian\}@stanford.edu}}
			\and
			Jerry Li\thanks{Microsoft Research, {\tt jerrl@microsoft.com}}
			\and
			Kevin Tian\footnotemark[1] 
		}
		
		\maketitle

\abstract{
We develop two methods for the following fundamental statistical task: given an $\eps$-corrupted set of $n$ samples from a $d$-dimensional sub-Gaussian distribution, return an approximate top eigenvector of the covariance matrix. Our first robust PCA algorithm runs in polynomial time, returns a $1 - O(\eps\log\eps^{-1})$-approximate top eigenvector, and is based on a simple iterative filtering approach. Our second, which attains a slightly worse approximation factor, runs in nearly-linear time and sample complexity under a mild spectral gap assumption. These are the first polynomial-time algorithms yielding non-trivial information about the covariance of a corrupted sub-Gaussian distribution without requiring additional algebraic structure of moments. As a key technical tool, we develop the first width-independent solvers for Schatten-$p$ norm packing semidefinite programs, giving a $(1 + \eps)$-approximate solution in $O(p\log(\tfrac{nd}{\eps})\eps^{-1})$ input-sparsity time iterations (where $n$, $d$ are problem dimensions).
}
 		
	\end{titlepage}

\section{Introduction}
\label{sec:intro}

We study two natural, but seemingly unrelated, problems in high dimensional robust statistics and continuous optimization respectively. As we will see, these problems have an intimate connection.

\textbf{Problem 1: Robust sub-Gaussian principal component analysis.} We consider the following statistical task, which we call \emph{robust sub-Gaussian principal component analysis} (PCA). Given samples $X_1, \ldots, X_n$ from sub-Gaussian\footnote{See Section~\ref{sec:prelims} for a formal definition.} distribution $\dist$ with covariance $\covar$, an $\eps$ fraction of which are arbitrarily corrupted, the task asks to output unit vector $u$ with $u^\top \covar u \ge (1 - \gamma)\norm{\covar}_{\infty}$\footnote{Throughout we use $\norm{\mm}_p$ to denote the Schatten $p$-norm (cf.\ Section~\ref{sec:prelims} for more details).} for tolerance $\gamma$. Ergo, the goal is to robustly return a $(1 - \gamma)$-approximate top eigenvector of the covariance of sub-Gaussian $\dist$. This is the natural extension of PCA to the robust statistics setting.

There has been a flurry of recent work on efficient algorithms for robust statistical tasks, e.g.\ covariance estimation and PCA. From an information-theoretic perspective, sub-Gaussian concentration suffices for robust covariance estimation. Nonetheless, to date all polynomial-time algorithms achieving nontrivial guarantees on covariance estimation of a sub-Gaussian distribution (including PCA specifically) in the presence of adversarial noise require additional algebraic structure. For instance, sum-of-squares certifiably bounded moments have been leveraged in polynomial time covariance estimation \cite{hopkins2018mixture, kothari2018robust}; however, this is a stronger assumption than sub-Gaussianity.

In many applications (see discussion in \cite{diakonikolas2017being}), the end goal of covariance estimation is PCA. Thus, a natural question which relaxes robust covariance estimation is: can we robustly estimate the top eigenvector of the covariance $\covar$, assuming only sub-Gaussian concentration? Our work answers this question affirmatively via two incomparable algorithms.
The first achieves $\gamma = O(\eps \log \eps^{-1})$ in polynomial time; the second achieves $\gamma = O(\sqrt{\eps\log\eps^{-1} \log d})$ in nearly-linear time under a mild gap assumption on $\covar$. Moreover, both methods have nearly-optimal sample complexity.

\textbf{Problem 2: Width-independent Schatten packing.} 
We consider a natural generalization of packing semidefinite programs (SDPs) which we call \emph{Schatten packing}. Given symmetric positive semidefinite $\ma_1, \ldots, \ma_n$ and parameter $p \geq 1$, a Schatten packing SDP asks to solve the optimization problem
\begin{equation}
\label{eq:packing}
\min \norm{\sum_{i \in[n]} w_i \ma_i}_p \text{ subject to } w \in \Delta^n.
\end{equation}
Here, $\norm{\mm}_p$ is the Schatten-$p$ norm of matrix $\mm$ and $\Delta^n$ is the probability simplex (see Section~\ref{sec:prelims}). When $p = \infty$, \eqref{eq:packing} is the well-studied (standard) packing SDP objective \cite{JainY11, Allen-ZhuLO16, PengTZ16}, which asks to find the most spectrally bounded convex combination of packing matrices. For smaller $p$, the objective encourages combinations more (spectrally) uniformly distributed over directions. 

The specialization of \eqref{eq:packing} to diagonal matrices is a smooth generalization of packing linear programs, previously studied in the context of fair resource allocation \cite{MarasevicSZ16, DiakonikolasFO18}. For the $\ell_\infty$ case of \eqref{eq:packing}, packing SDPs have the desirable property of admitting ``width-independent'' approximation algorithms via exploiting positivity structure. Specifically, width-independent solvers obtain multiplicative approximations with runtimes independent or logarithmically dependent on size parameters of the problem. This is a strengthening of additive notions of approximation typically used for approximate semidefinite programming. Our work gives the first width-independent solver for Schatten packing.

\subsection{Previous work}

\textbf{Learning with adversarial outliers.} The study of estimators robust to a small fraction of adversarial outliers dates back to foundational work, e.g.\ \cite{huber1964robust, tukey1975mathematics}. Following more recent work \cite{lai2016agnostic, DiakonikolasKKLMS19}, there has been significant interest in efficient, robust algorithms for statistical tasks in high-dimensional settings. We focus on methods robustly estimating covariance properties here, and defer a thorough discussion of the (extensive) robust statistics literature to \cite{steinhardt2018robust, li2018principled, diakonikolas2019recent}.

There has been quite a bit of work in understanding and giving guarantees for robust covariance estimation where the uncorrupted distribution is exactly Gaussian \cite{diakonikolas2017being, diakonikolas2018robustly, DiakonikolasKKLMS19, cheng2019faster}. These algorithms strongly use relationships between higher-order moments of Gaussian distributions via Isserlis' theorem. Departing from the Gaussian setting, work of \cite{lai2016agnostic} showed that if the distribution is an affine transformation of a 4-wise independent distribution, robust covariance estimation is possible. This was extended by \cite{kothari2018robust}, which also assumed nontrivial structure in the moments of the distribution, namely that sub-Gaussianity was certifiable via the sum-of-squares proof system. To the best of our knowledge it has remained open to give nontrivial guarantees for robust estimation of any covariance properties under minimal assumptions, i.e.\ sub-Gaussian concentration.

All aforementioned algorithms also yield guarantees for robust PCA, by applying a top eigenvector method to the learned covariance. However, performing robust PCA via the intermediate covariance estimation step is lossy, both statistically and computationally. From a statistical perspective, $\Omega(d^2)$ samples are necessary to learn the covariance of a $d$-dimensional Gaussian in Frobenius norm (and for known efficient algorithms for spectral norm error \cite{diakonikolas2017statistical}); in contrast, $O(d)$ samples suffice for (non-robust) PCA. Computationally, even when the underlying distrubution is exactly Gaussian, the best-known covariance estimation algorithms run in time $\Omega(d^{3.25})$; algorithms working in more general settings based on the sum-of-squares approach require much more time. In contrast, the power method for PCA in a $d \times d$ matrix takes time $\tO(d^2)$\footnote{We say $g = \tO(f)$ if $g = O(f \log^c f)$ for some constant $c > 0$.}. Motivated by this, our work initiates the direct study of robust PCA, which is often independently interesting in applications.

We remark there is another problem termed ``robust PCA'' in the literature, e.g.\ \cite{candes2011robust}, under a different generative model. We defer a detailed discussion to \cite{diakonikolas2017being}, which experimentally shows that algorithms from that line of work do not transfer well to our corruption model.

\textbf{Width-independent iterative methods.} Semidefinite programming (SDP) and its linear programming specialization are fundamental computational tasks, with myriad applications in learning, operations research, and computer science. Though general-purpose polynomial time algorithms exist for SDPs (\cite{NesterovN94}), in practical settings in high dimensions, approximations depending linearly on input size and polynomially on error $\eps$ are sometimes desirable. To this end, approximation algorithms based on entropic mirror descent have been intensely studied \cite{WarmuthK06, AroraK16, GarberHM15, Allen-ZhuL17b, CarmonDST19}, obtaining $\eps$ additive approximations to the objective with runtimes depending polynomially on $\rho/\eps$, where $\rho$ is the ``width'', the largest spectral norm of a constraint.

For structured SDPs, stronger guarantees can be obtained in terms of width. Specifically, several algorithms developed for packing SDPs (\eqref{eq:packing} with $p = \infty$) yield $(1 + \eps)$-\emph{multiplicative} approximations to the objective, with \emph{logarithmic} dependence on width \cite{JainY11, PengTZ16, Allen-ZhuLO16, JambulapatiLLPT20}. As $\rho$ upper bounds objective value in this setting, in the worst case runtimes of width-dependent solvers yielding $\eps\rho$-additive approximations have similar dependences as width-independent counterparts. Width-independent solvers simultaneously yield stronger multiplicative bounds at all scales of objective value, making them desirable in suitable applications. In particular, $\ell_\infty$ packing SDPs have found great utility in robust statistics algorithm design \cite{ChengG18, ChengDG19, cheng2019faster, DespersinL19}.

Beyond $\ell_\infty$ packing, width-independent guarantees in the SDP literature are few and far between; to our knowledge, other than the covering and mixed solvers of \cite{JambulapatiLLPT20}, ours is the first such guarantee for a broader family of objectives\footnote{In concurrent and independent work, \cite{CherapanamjeriMY20} develops width-independent solvers for Ky-Fan packing objectives, a different notion of generalization than the Schatten packing objectives we consider.}. Our method complements analogous $\ell_p$ extensions in the width-dependent setting, e.g.\ \cite{ZhuLO15}, as well as width-independent solvers for $\ell_p$ packing linear programs \cite{MarasevicSZ16, DiakonikolasFO18}. We highlight the fair packing solvers of \cite{MarasevicSZ16, DiakonikolasFO18}, motivated by problems in equitable resource allocation, which further solved $\ell_p$ packing variants for $p \not\in [1, \infty)$. We find analogous problems in semidefinite settings interesting, and defer to future work.

\subsection{Our results}
\label{ssec:ourresults}

\textbf{Robust sub-Gaussian principal component analysis.} We give two algorithms for robust sub-Gaussian PCA\footnote{We follow the distribution and corruption model described in Assumption~\ref{assume:corruption}.}. 
Both are near-sample optimal, polynomial-time, and assume only sub-Gaussianity.
The first is via a simple filtering approach, as summarized here (and developed in Section~\ref{sec:filtering}).

\begin{restatable}{theorem}{restatepolyfinal}
	\label{thm:poly-final}
	Under Assumption~\ref{assume:corruption}, let $\delta \in [0, 1]$, and $n = \Omega\Par{\tfrac{d + \log\delta^{-1}}{(\eps\log\eps^{-1})^2}}$. Algorithm~\ref{alg:filter} runs in time $O(\tfrac{nd^2}{\eps} \log\tfrac{n}{\delta\eps}\log\tfrac{n}{\delta})$, and outputs $u$ with $u^\top \covar u > (1 - C^\star \eps \log \eps^{-1}) \| \covar \|_\infty$, for $C^\star$ a fixed multiple of parameter $c$ in Assumption~\ref{assume:corruption}, with probability at least $1 - \delta$.
\end{restatable}

Our second algorithm is more efficient under mild conditions, but yields a worse approximation $1 - \gamma$ for $\gamma = O(\sqrt{\eps\log\eps^{-1}\log d})$. Specifically, if there are few eigenvalues of $\covar$ larger than $1 - \gamma$, our algorithm runs in nearly-linear time. Note that if there are many eigenvalues above this threshold, then the PCA problem itself is not very well-posed; our algorithm is very efficient in the interesting setting where the approximate top eigenvector is identifiable. We state our main algorithmic guarantee here, and defer details to Section~\ref{sec:subspace}.
\begin{theorem}
	\label{thm:main-robust}
	Under Assumption~\ref{assume:corruption}, let $\delta \in [0, 1]$, $n =\nomega$, $\gamma = C \sqrt{\epsilon \log \eps^{-1} \log d}$,  for $C$ a fixed multiple of parameter $c$ from Assumption~\ref{assume:corruption}, and let $t \in [d]$ satisfy $\covar_{t + 1} < (1 - \gamma)\norm{\covar}_\infty$. Algorithm~\ref{alg:robust-pca} outputs a unit vector $u \in \R^d$ with $u^\top \covar u \geq (1 - \gamma) \| \covar\|_\infty$ in time $\tO(\tfrac{nd}{\eps^{4.5}} + \tfrac{ndt}{\eps^{1.5}})$.
\end{theorem}
We remark that $\Omega (d \eps^{-2})$ samples are necessary for a $(1 - \eps)$-approximation to the top eigenvector of $\covar$ via uncorrupted samples from $\Nor(0, \covar)$, so our first method is sample-optimal, as is our second up to a $\tO(\eps^{-1})$ factor.

\textbf{Width-independent Schatten packing.} Our second method crucially requires an efficient solver for Schatten packing SDPs. We demonstrate that Schatten packing, i.e.~\eqref{eq:packing} for arbitrary $p$, admits width-independent solvers. We state an informal guarantee, and defer details to Section~\ref{sec:packing}.
	\begin{theorem}
	\label{thm:main-informal-1}
		Let $\{\ma_i\}_{i \in [n]} \in \PSD^d$, and $\eps > 0$. There is an algorithm taking $O(\tfrac{p \log (\frac{nd}{\eps})}{\eps})$ iterations, returning a $1 + \eps$ multiplicative approximation to the problem~\eqref{eq:packing}. For odd $p$, each iteration can be implemented in time nearly-linear in the number of nonzeros amongst all $\{\ma_i\}_{i \in [n]}$.
	\end{theorem}

%

\section{Preliminaries}
\label{sec:prelims}

\textbf{General notation.} $[n]$ denotes the set $1 \le i \le n$. Applied to a vector, $\norm{\cdot}_p$ is the $\ell_p$ norm; applied to a symmetric matrix, $\norm{\cdot}_p$ is the Schatten-$p$ norm, i.e.\ the $\ell_p$ norm of the spectrum. The \emph{dual norm} of $\ell_p$ is $\ell_q$ for $q = \tfrac{p}{p - 1}$; when $p = \infty$, $q = 1$. $\Delta^n$ is the $n$-dimensional simplex (subset of positive orthant with $\ell_1$-norm $1$) and we define $\fS^n_{\veps} \subseteq \Delta^n$ to be the truncated simplex:
\begin{equation}\label{eq:truncsimp}\fS^n_{\veps} \defeq \Brace{w \in \R^n_{\ge 0} \;\Bigg| \; \norm{w}_1 = 1,\; w \le \frac{1}{n (1 - \veps)} \text{ entrywise}}.\end{equation}

\textbf{Matrices.} $\Sym^d$ is $d \times d$ symmetric matrices, and $\PSD^d$ is the positive semidefinite subset. $\id$ is the identity of appropriate dimension. $\lmax$, $\lmin$, and $\Tr$ are the largest and smallest eigenvalues and trace of a symmetric matrix. For $\mm, \mn \in \Sym^d$, $\inprod{\mm}{\mn} \defeq \Tr\Par{\mm\mn}$ and we use the Loewner order $\preceq$, ($\mm \preceq \mn$ iff $\mn - \mm \in \PSD^d$). The seminorm of $\mm \succeq 0$ is $\norm{v}_{\mm} \defeq \sqrt{v^\top\mm v}$. 

\begin{fact}\label{fact:trprod}
For $\ma$, $\mb$ with compatible dimension, $\Tr(\ma\mb) = \Tr(\mb\ma)$. For $\mm, \mn \in \PSD^d$, $\inprod{\mm}{\mn} \ge 0$.
\end{fact}

\begin{fact}\label{fact:qnormcert}
We have the following characterization of the Schatten-$p$ norm: for $\mm \in \Sym^d$, and $q = \tfrac{p}{p - 1}$,
\[\norm{\mm}_p = \sup_{\substack{\mn \in \Sym^d,\; \norm{\mn}_q = 1}} \inprod{\mn}{\mm}.\]
For $\mm = \sum_{j \in [d]} \lam_i v_i v_i^\top$, the satisfying $\mn$ is $\tfrac{\sum_{j \in [d]} \pm \lam_i^{p - 1} v_i v_i^\top}{\norm{\mm}_p^{p - 1}}$, so $\mn\mm$ has spectrum $\tfrac{|\lam|^p}{\norm{\mm}_p^{p - 1}}$. 
\end{fact}

\textbf{Distributions.} We denote drawing vector $X$ from distribution $\dist$ by $X \sim \dist$, and the covariance $\covar$ of $\dist$ is $\E_{X \sim \dist}\Brack{XX^\top}$. We say scalar distribution $\dist$ is $\gamma^2$-sub-Gaussian if $\E_{X \sim \dist}[X] = 0$ and 
\[\E_{X \sim \dist}\Brack{\exp\Par{tX}} \le \exp\Par{\frac{t^2\gamma^2}{2}} \; \forall t \in \R.\]
Multivariate $\dist$ has sub-Gaussian proxy $\mprox$ if its restriction to any unit $v$ is $\norm{v}_{\mprox}^2$-sub-Gaussian, i.e.\
\begin{equation}\label{eq:subg_def}\E_{X \sim \dist}\Brack{\exp\Par{tX^\top v}} \le \exp\Par{\frac{t^2\norm{v}_{\mprox}^2}{2}} \text{ for all }\norm{v}_2 = 1,\; t \in \R.\end{equation}

We consider the following standard model for gross corruption with respect to distribution a $\dist$. 

\begin{assumption}[Corruption model, see~\cite{DiakonikolasKKLMS19}]\label{assume:corruption}
Let $\dist$ be a mean-zero distribution on $\R^d$ with covariance $\covar$ and sub-Gaussian proxy $\mprox \preceq c\covar$ for a constant $c$. Denote by index set $G'$ with $|G'| = n$ a set of (uncorrupted) samples $\{X_i\}_{i \in G'} \sim \dist$. An adversary arbitrarily replaces $\eps n$ points in $G'$; we denote the new index set by $[n] = B \cup G$, where $B$ is the (unknown) set of points added by an adversary, and $G \subseteq G'$ is the set of points from $G'$ that were not changed. 
\end{assumption}

As we only estimate covariance properties, the assumption that $\dist$ is mean-zero only loses constants in problem parameters, by pairing samples and subtracting them (cf.\ \cite{DiakonikolasKKLMS19}, Section 4.5.1). 

\section{Robust sub-Gaussian PCA via filtering}
\label{sec:filtering}

In this section, we sketch the proof of Theorem~\ref{thm:poly-final}, which gives guarantees on our filtering algorithm for robust sub-Gaussian PCA. This algorithm obtains stronger statistical guarantees than Theorem~\ref{thm:main-robust}, at the cost of super-linear runtime; the algorithm is given as Algorithm~\ref{alg:filter}. Our analysis stems largely from concentration facts about sub-Gaussian distributions, as well as the following (folklore) fact regarding estimation of variance along any particular direction.

\begin{restatable}{lemma}{univariate}\label{lem:univariate}
	Under Assumption~\ref{assume:corruption}, let $\delta \in [0, 1]$, $n = \Omega \left(\tfrac{\log \delta^{-1}}{(\eps \log \eps^{-1})^{2}} \right)$, and $u \in \R^d$ be a fixed unit vector. Algorithm~\ref{alg:univariate}, $\OneD$, takes input $\{X_i\}_{i \in [n]}$, $u$, and $\eps$, and outputs $\sigma_u^2$ with $|u^\top \covar u - \sigma_u^2| < C u^\top \covar u \cdot \epsilon \log \eps^{-1}$ with probability at least $1 - \delta$, and runs in time
	$O(nd + n \log n)$, for $C$ a fixed multiple of the parameter $c$ in Assumption~\ref{assume:corruption}.
\end{restatable}

In other words, we show that using corrupted samples, we can efficiently estimate a $1 + O(\eps\log\eps^{-1})$-multiplicative approximation of the variance of $\dist$ in any unit direction\footnote{Corollary~\ref{corr:oned-independence} gives a slightly stronger guarantee that reusing samples does not break dependencies of $u$.}. This proof is deferred to Appendix~\ref{app:filtering} for completeness. Algorithm~\ref{alg:filter} combines this key insight with a soft filtering approach, suggested by the following known structural fact found in previous work (e.g.\ Lemma A.1 of \cite{dong2019quantum}, see also \cite{steinhardt2017resilience,steinhardt2018robust}).

\begin{lemma}
	\label{lem:uni-filter}
	Let $\{a_i\}_{i \in [m]}$, $\{w_i\}_{i \in [m]}$ be sets of nonnegative reals, and $a_{\max} = \max_{i \in [m]} a_i$. Define
	$
	w'_i = \Par{1 - \frac{a_i}{a_{\max}}} w_i$, for all $i \in [m]$.
	Consider any disjoint partition $I_B$, $I_G$ of $[m]$ with
	$
	\sum_{i \in I_B} w_i a_i > \sum_{i \in I_G} w_i a_i.
	$
	Then, $\sum_{i \in I_B} w_i - w_i' > \tfrac{1}{2 a_{\max}} \sum_{i \in [m]} w_i a_i > \sum_{i \in I_G} w_i - w_i'$.
\end{lemma}

Our Algorithm~\ref{alg:filter}, $\PCAFilter$, takes as input a set of corrupted samples $\{X_i\}_{i \in [n]}$ following Assumption~\ref{assume:corruption} and the corruption parameter $\eps$. At a high level, it initializes a uniform weight vector $w^{(0)}$, and iteratively operates as follows (we denote by $\mm(w)$ the empirical covariance $\sum_{i \in [n]} w_i X_i X_i^\top$).
\begin{enumerate}
	\item $u_t \gets $ approximate top eigenvector of $\mm(w^{(t - 1)})$ via power iteration.
	\item Compute $\sigma_t^2 \gets \OneD(\{X_i\}_{i \in [n]}, u_t, \eps)$.
	\item If $\sigma_t^2 > (1 - O(\eps\log\eps^{-1})) \cdot u_t^\top \mm(w^{(t - 1)}) u_t$, then terminate and return $u_t$.
	\item Else:
	\begin{enumerate}
		\item Sort indices $i \in [n]$ by $a_i \gets \inprod{u_t}{X_i}^2$, with $a_1$ smallest.
		\item Let $\ell \le i \le n$ be the smallest set for which $\sum_{i = \ell}^n w_i \ge 2\eps$, and apply the downweighting procedure of Lemma~\ref{lem:uni-filter} to this subset of indices.
	\end{enumerate}
\end{enumerate}

The analysis of Algorithm~\ref{alg:filter} then proceeds in two stages.

\paragraph{Monotonicity of downweighting.} We show the invariant criteria for Lemma~\ref{lem:uni-filter} (namely, that for the set $\ell \le i \le n$ in every iteration, there is more spectral mass on bad points than good) holds inductively for our algorithm. Specifically, lack of termination implies $\mm(w^{(t - 1)})$ puts significant mass on bad directions, which combined with concentration of good directions yields the invariant. The details of this argument can be found as Lemma~\ref{lem:inv-holds}.

\paragraph{Roughly uniform weightings imply approximation quality.} As Lemma~\ref{lem:uni-filter} then applies, the procedure always removes more mass from bad points than good, and thus can only remove at most $2\eps$ mass total by the corruption model. Thus, the weights $w^{(t)}$ are always roughly uniform (in $\fS_{O(\eps)}^n$), which by standard concentration facts (see Appendix~\ref{app:concentration}) imply the quality of the approximate top eigenvector is good. Moreover, the iteration count is bounded by roughly $d$ because whenever the algorithm does not terminate, enough mass is removed from large spectral directions. Combining with the termination criteria imply that when a vector is returned, it is a close approximation to the top direction of $\covar$. Details can be found as Lemma~\ref{lem:output-correct} and in the proof of Theorem~\ref{thm:poly-final}.

\section{Schatten packing}
\label{sec:packing}

\subsection{Mirror descent interpretation of \cite{MahoneyRWZ16}}
\label{ssec:mrwzonsimplex}
We begin by reinterpreting the \cite{MahoneyRWZ16} solver, which achieves the state-of-the-art parallel runtime for packing LPs\footnote{The \cite{MahoneyRWZ16} solver also generalizes to covering and mixed objectives; we focus on packing in this work.}. An $(\ell_\infty)$ packing LP algorithm solves the following decision problem.\footnote{Packing linear programs are sometimes expressed as the optimization problem $\max_{x \ge 0, \ma x \le \1} \norm{x}_1$, similarly to \eqref{eq:packing}; these problems are equivalent up to a standard binary search, see e.g.\ discussion in \cite{JambulapatiLLPT20}.}.

\begin{problem}[$\ell_\infty$ packing linear program]
\label{problem:packing}
Given entrywise nonnegative $\ma \in \mathbb{R}_{\geq 0}^{d \times n}$, either find primal solution $x \in \Delta^n$ with $\norm{\ma x}_\infty \leq 1+\epsilon$ or dual solution $y \in \Delta^d$ with $\ma^\top y \geq (1-\epsilon) \1$. 
\end{problem}

\begin{algorithm}
	\caption{$\PackingLP(\ma,\epsilon)$}
	\begin{algorithmic}[1]\label{alg:mrwz}
	\STATE \textbf{Input:} $\ma \in \R^{d \times n}_{\ge 0}$, $\eps \in [0, \thalf]$
	\STATE $K \gets \tfrac{3 \log (d)}{\epsilon}$, $\eta \gets K^{-1}$, $T \gets \tfrac{4 \log(d)\log(nd/\eps)}{\epsilon^2}$
	\STATE $[w_0]_i \gets \tfrac{\eps}{n^2 d}$ for all $i \in [n]$, $z \gets \mzero$, $t \gets 0$
	\WHILE{$\ma w_t \leq K \1$, $\norm{w_t}_1 \le K$}
		\STATE $v_t \gets \tfrac{\exp(\ma w_t)}{\norm{\exp(\ma w_t)}_1}$
		\STATE $g_t \gets \max(0, \1 - \ma^\top v_t)$ entrywise
		\STATE $w_{t + 1} \gets w_t \circ (1 + \eta g_t)$, $z \gets z + v_t$, $t \gets t + 1$
		\IF{$t \geq T$}
			\RETURN{$y \gets \frac{1}{T} z$}
		\ENDIF
	\ENDWHILE
	\RETURN{ $x \gets \frac{ w_t }{\norm{w_t}_1}$ }
	\end{algorithmic}
\end{algorithm}
The following result is shown in \cite{MahoneyRWZ16}.
\begin{proposition}
\label{prop:mrwz}
$\PackingLP$ (Algorithm~\ref{alg:mrwz}) solves Problem~\ref{problem:packing} in $O(\nnz(\ma) \cdot \tfrac{\log(d)\log(nd/\eps)}{\epsilon^2})$ time. 
\end{proposition}

Oour interpretation of the analysis of \cite{MahoneyRWZ16}, combines two ingredients: a potential argument and mirror descent, which yields a dual feasible point if $\norm{w_t}_1$ did not grow sufficiently.

\textbf{Potential argument.} The potential used by \cite{MahoneyRWZ16} is $\log(\sum_{j \in [d]}\exp([\ma w_t]_j)) - \norm{w_t}_1$, well-known to be a $O(\log d)$-additive approximation of $\norm{\ma w_t}_\infty - \norm{w_t}_1$. As soon as $\norm{\ma w_t}_\infty$ or $\norm{w_t}_1$ reaches the scale $O(\tfrac{\log d}{\eps})$, by nonnegativity this becomes a multiplicative guarantee, motivating the setting of threshold $K$. To prove the potential is monotone, \cite{MahoneyRWZ16} uses step size $K^{-1}$ and a Taylor approximation; combining with the termination condition yields the desired claim.

\textbf{Mirror descent.} To certify that $w_t$ grows sufficiently (e.g.\ the method terminates in few iterations, else dual feasibility holds), we interpret the step $w_{t + 1} \gets w_t \circ (1 + \eta g_t)$ as approximate entropic mirror descent. Specifically, we track the quantity $\sum_{0 \le t < T} \inprod{\eta g_t}{u}$, and show that if $\norm{w_t}_1$ has not grown sufficiently, then it must be bounded for every $u \in \Delta^n$, certifying dual feasibility. Formally, for any $g_t$ sequence and $u \in \Delta^n$, we show
\[O(\log(nd/\eps)) + \log\Par{\frac{\norm{w_T}_1}{\norm{w_0}_1}}\ge \sum_{0 \le t < T} \inprod{\eta g_t}{u} \ge \eta\sum_{0 \le t < T} \inprod{\1 - \ma^\top v_t}{u}. \]
The last inequality followed by $g_t$ being an upwards truncation. If $\norm{w_T}_1$ is bounded (else, we have primal feasibility), we show the entire above expression is bounded $O(\log\tfrac{nd}{\eps})$ for any $u$. Thus, by setting $T = O(\tfrac{\log(nd/\eps)}{\eta\eps})$ and choosing $u$ to be each coordinate indicator, it follows that the average of all $v_t$ is coordinatewise at least $1 - \eps$, and solves Problem~\ref{problem:packing} as a dual solution. 

Our $g_t$ is chosen as the (truncated) gradient of the function used in the potential analysis, so its form allows us to interpret dual feasibility (e.g.\ $v_t$ has $\ell_1$ norm 1 and is a valid dual point). Our analysis patterns standard mirror descent, complemented by side information which says that lack of a primal solution can transform a regret guarantee into a feasibility bound. We apply this framework to analyze $\ell_p$ variants of Problem~\ref{problem:packing}, via different potentials; nonetheless, our proofs are quite straightforward upon adopting this perspective, and we believe it may yield new insights for instances with positivity structure.

\subsection{$\ell_p$-norm packing linear programs}
\label{sec:lp-lp}
In this section, we give a complete self-contained example of the framework proposed in Section~\ref{ssec:mrwzonsimplex}, for approximately solving $\ell_p$ norm packing linear programs. Specifically, we now consider the generalization of Problem~\ref{problem:packing} to $\ell_p$ norms; throughout, $q = \tfrac{p}{p - 1}$ is the dual norm.

\begin{problem}[$\ell_p$ packing linear program]
\label{problem:lp-packing}
Given entrywise nonnegative $\ma \in \R_{\geq 0}^{d \times n}$, either find primal solution $x \in \Delta^n$ with $\norm{\ma x}_p \leq 1+\epsilon$ or dual solution $y \in \R^d_{\ge 0}, \norm{y}_q = 1$ with $\ma^\top y \geq (1-\epsilon) \1$.
\end{problem}

For $p = \tfrac{\log d}{\eps}$, Problem~\ref{problem:lp-packing} recovers Problem~\ref{problem:packing} up to constants as $\ell_p$ multiplicatively approximates $\ell_\infty$ by $1 + \eps$. We now state our method for solving Problem~\ref{problem:lp-packing} as Algorithm~\ref{alg:lp-lp}.
\begin{algorithm}[ht!]
	\caption{$\lppacking(\ma, \eps, p)$}
	\begin{algorithmic}[1]\label{alg:lp-lp}
		\STATE \textbf{Input:} $\ma \in \R_{\ge 0}^{d \times n}, \eps \in [0, \thalf], p \ge 2$
		\STATE $\eta \gets p^{-1}, T \gets \tfrac{4p\log(\frac{nd}{\eps}) }{\epsilon}$
		\STATE $[w_0]_i \gets \tfrac{\eps}{n^2 d}$ for all $i \in [n]$, $z \gets \mzero$, $t \gets 0$
		\WHILE{$\norm{w_t}_1 \leq \epsilon^{-1}$}	
		\STATE $g_t \gets \max(0, \1 - \ma^\top (v_t)^{p-1})$ entrywise, for $v_t \gets \tfrac{\ma w_t}{\norm{\ma w_t}_p}$
		\STATE $w_{t + 1} \gets w_t \circ (1 + \eta g_t)$, $z \gets z + (v_t)^{p-1}$, $t \gets t+1$
		\IF{$t \geq T$}
		\RETURN{$y  = \frac{z}{\norm{z}_q}$}
		\ENDIF
		\ENDWHILE
		\RETURN{ $x = \frac{w_t}{\norm{w_t}_1}$ }
	\end{algorithmic}
\end{algorithm}

Other than changing parameters, the only difference from Algorithm~\ref{alg:mrwz} is that $v$ is a point with unit $\ell_q$ norm induced by the gradient of our potential $\Phi_t$. We state our main potential fact, whose proof is based straightforwardly on Taylor expanding $\norm{\cdot}_p$, and deferred to Appendix~\ref{app:pack} for brevity.

\begin{restatable}{lemma}{restatelplppotential}
\label{lemma:lp-potential}
In all iterations $t$ of Algorithm~\ref{alg:lp-lp}, defining $\Phi_t \defeq \norm{\ma w_t}_p - \norm{w_t}_1$, $\Phi_{t+1} \leq \Phi_t$. 
\end{restatable}

We now prove our main result, which leverages the potential bound following the framework of Section~\ref{ssec:mrwzonsimplex}. In the proof, we assume that entries of $\ma$ are bounded by $n\eps^{-1}$; this does not incur more loss than a constant multiple of $\eps$ in the guarantees, and a proof can be found as Lemma~\ref{lem:assumeabounded}.

\begin{restatable}{theorem}{restatewgrowth}\label{thm:wgrowth}
Algorithm~\ref{alg:lp-lp} runs in time $O(\nnz(\ma) \cdot \tfrac{p\log(nd/\eps)}{\eps})$. Further, its output solves Problem~\ref{problem:lp-packing}.
\end{restatable}
\begin{proof}
	The runtime follows from Line 7 (each iteration cost is dominated by multiplication through $\ma$), so we prove correctness. Define potential $\Phi_t$ as in Lemma~\ref{lemma:lp-potential}, and note that as $w_0 = \tfrac{\eps}{n^2d}\1$,
	\[\Phi_0 \le \norm{\ma w_0}_p \le \frac{1}{n}\norm{\1}_p \le 1.\]
	The second inequality followed from our assumption on $\ma$ entry sizes (Lemma~\ref{lem:assumeabounded}). If Algorithm~\ref{alg:lp-lp} breaks out of the while loop of Line 4, we have by Lemma~\ref{lemma:lp-potential} that for $x$ returned on Line 11,
	\[\norm{\ma w_t}_p - \norm{w_t}_1 \le 1 \implies \norm{\ma x}_p \le \frac{1 + \norm{w_t}_1}{\norm{w_t}_1} \le 1 + \eps.\]
	Thus, primal feasibility is always correct. We now prove correctness of dual feasibility. First, let $V_x(u) = \sum_{i \in [n]} u_i \log (\tfrac{u_i}{x_i})$ be the Kullback-Leibler divergence from $x$ to $u$, for $x$, $u \in \Delta^d$. Define the normalized points $x_t = \tfrac{w_t}{\norm{w_t}_1}$ in each iteration. Expanding definitions,
	\begin{equation}\label{eq:telescope_lp}
	\begin{aligned}
	V_{x_{t + 1}}(u) - V_{x_t}(u) &= \sum_{i \in [n]} u_i \log \frac{[x_t]_i}{[x_{t + 1}]_i}\\
	&= \sum_{i \in [n]} u_i\left(\log\Par{\frac{\norm{w_{t + 1}}_1}{\norm{w_t}_1}} + \log\left(\frac{1}{1+\eta[g_t]_i}\right)\right) \\
	&\le -\eta(1-\eta)\inprod{g_t}{u} + \log\left(\frac{\norm{w_{t + 1}}_1}{\norm{w_t}_1}\right).
	\end{aligned}
	\end{equation}
	The only inequality used the bounds, for $g, \eta \in [0, 1]$,
	\[\log\left(\frac{1}{1 + \eta g}\right) \le g\log\left(\frac{1}{1 + \eta}\right) \le -\eta(1 - \eta)g. \]
	Telescoping \eqref{eq:telescope_lp} over all $T$ iterations, and using $V_{x_0}(u) \le \log n$ for all $u \in \Delta^n$ since $x_0$ is uniform, we have that whenever Line 4 is not satisfied before the check on Line 7 (i.e.\ $t \ge T$),
	\begin{equation}\label{eq:upperboundtelescope}\eta(1 - \eta)\sum_{0 \le t < T} \inprod{g_t}{u} \le \log\Par{\frac{\norm{w_T}_1}{\norm{w_0}_1}} + V_{x_0}(u) \le \log\Par{\frac{nd}{\eps^2}} + \log n \le 2\log\Par{\frac{nd}{\eps}}.\end{equation}
	The last inequality used $\norm{w_T}_1 \le \eps^{-1}$ by assumption, and $\norm{w_0}_1 = \tfrac{\eps}{nd}$. Next, since each $g_t \ge \1 - \ma^\top(v_t)^{p - 1}$ entrywise, defining $\bar{z} = \tfrac{z}{T}$,
	\begin{equation}\label{eq:lowerboundsum}\sum_{0 \le t < T}\inprod{g_t}{u} \ge \sum_{0 \le t < T} \inprod{\1 - \ma^\top (v_t)^{p - 1}}{u} = T\inprod{\1 - \ma^\top \bar{z}}{u}, \text{ for all } u \in \Delta^n.\end{equation}
	Combining \eqref{eq:upperboundtelescope} and \eqref{eq:lowerboundsum}, and rearranging, yields by definition of $T$,
	\[\inprod{\1 - \ma^\top \bar{z}}{u} \le \frac{2\log(\frac{nd}{\eps})}{T\eta(1 - \eta)} \le \frac{4p\log(\frac{nd}{\eps})}{T} \le \eps \implies \ma^\top \bar{z} \ge 1 - \eps \text{ entrywise.}\]
	The last claim follows by setting $u$ to be each coordinate-sparse simplex vector. Finally, since $\tfrac{\bar{z}}{\norm{\bar{z}}_q} = \tfrac{z}{\norm{z}_q}$, to show that $y$ is a correct dual solution to Problem~\ref{problem:lp-packing} it suffices to show $\norm{\bar{z}}_q \le 1$. This follows as $\bar{z}$ is an average of the $(v_t)^{p - 1}$, convexity of $\ell_q$ norms, and that for all $t$, 
	\[\norm{(v_t)^{p - 1}}_q^q = \sum_{j \in [d]} v_t^p = \sum_{j \in [d]} \frac{\Brack{\ma w_t}_j^p}{\norm{\ma w_t}_p^p} = 1.\]
\end{proof}

\subsection{Schatten-norm packing semidefinite programs}
\label{sec:lp-sdp}
We generalize Algorithm~\ref{alg:lp-lp} to solve Schatten packing semidefinite programs, which we now define.
\begin{problem}\label{problem:sdp-packing}
	Given $\{\ma_i\}_{i \in [n]} \in \PSD^d$, either find primal solution $x \in \Delta^n$ with $\norm{\sum_{i \in [n]} x_i \ma_i}_p \le 1 + \eps$ or dual solution $\my \in \PSD^d$, $\norm{\my}_q = 1$ with $\inprod{\ma_i}{\my} \ge 1 - \eps$ for all $i \in [n]$.
\end{problem}

\begin{algorithm}
	\caption{$\schattenpacking(\{\ma_i\}_{i \in [n]}, \eps, p)$}
	\begin{algorithmic}[1]\label{alg:lp-sdp}
		\STATE \textbf{Input:} $\{\ma_i\}_{i \in [n]} \in \PSD^d, \eps \in [0, \thalf], p \ge 2$
		\STATE $\eta \gets p^{-1}, T \gets \tfrac{4p\log(\frac{nd}{\eps}) }{\epsilon}$
		\STATE $[w_0]_i \gets \tfrac{\eps}{n^2 d}$ for all $i \in [n]$, $z \gets 0$
		\WHILE{$\norm{w_t}_1 \leq \epsilon^{-1}$}	
		\STATE $g_t \gets \max\Par{0, \1 - \inprod{\ma_i}{\mv_t^{p - 1}}}$ entrywise, for $\mv_t \gets \tfrac{\sum_{i \in [n]} [w_t]_i \ma_i}{\norm{\sum_{i \in [n]} [w_t]_i \ma_i}_p}$
		\STATE $w_{t + 1} \gets w_t \circ (1 + \eta g_t)$, $\mz \gets \mz + (\mv_t)^{p-1}$, $t \gets t+1$
		\IF{$t \geq T$}
		\RETURN{$\my  = \frac{\mz}{\norm{\mz}_q}$}
		\ENDIF
		\ENDWHILE
		\RETURN{ $x = \frac{w_t}{\norm{w_t}_1}$ }
	\end{algorithmic}
\end{algorithm}

We assume that $p$ is an odd integer for simplicity (sufficient for our applications), and leave for interesting future work the cases when $p$ is even or noninteger. The potential used in the analysis and an overall guarantee are stated here, and deferred to Appendix~\ref{app:pack}. The proofs are simple modifications of Lemma~\ref{lemma:lp-potential} and Theorem~\ref{thm:wgrowth} using trace inequalities (similar to those in \cite{JambulapatiLLPT20}) in place of scalar inequalities, as well as efficient approximation of quantities in Line 5 via the standard technique of Johnson-Lindestrauss projections.

\begin{restatable}{lemma}{restatesdppot}
\label{lem:sdp-potential}
In all iterations $t$ of Algorithm \ref{alg:lp-sdp}, defining $\Phi_t \defeq \norm{\sum_{i \in [n]} [w_t]_i \ma_i}_p - \norm{w_t}_1$,
$\Phi_{t+1} \leq \Phi_t$. 
\end{restatable}

\begin{restatable}{theorem}{restatewgrowthsdp}\label{thm:wgrowthsdp}
Let $p$ be odd. Algorithm~\ref{alg:lp-sdp} runs in $O(\tfrac{p\log(nd/\eps)}{\eps})$ iterations, and its output solves Problem~\ref{problem:sdp-packing}. Each iteration is implementable in $O(\nnz \cdot \tfrac{p\log(nd/\eps)}{\eps^2})$, where $\nnz$ is the number of nonzero entries amongst all $\{\ma_i\}_{i \in [n]}$, losing $O(\eps)$ in the quality of Problem~\ref{problem:sdp-packing} with probability $1 - \textup{poly}((nd/\eps)^{-1})$.
\end{restatable}

\subsection{Schatten packing with a $\ell_\infty$ constraint}

We remark that the framework outlined in Section~\ref{ssec:mrwzonsimplex} is flexible enough to handle mixed-norm packing problems. Specifically, developments in Section~\ref{sec:subspace} require the following guarantee.

\begin{restatable}{proposition}{restateboxconstrainedp}\label{prop:boxconstrainedp}
Following Theorem~\ref{thm:wgrowthsdp}'s notation, let $p$ be odd, $\{\ma_i\}_{i \in [n]} \in \PSD^d$, $0 < \eps = O(\alpha)$, and
\begin{equation}\label{eq:boxconstrainedschatten}\min_{\substack{x \in \Delta^n \\ \norm{x}_\infty \le \frac{1 + \alpha}{n}}} \norm{\alla(x)}_p = \textup{OPT}.
\end{equation}
for $\alla(x)\defeq \sum_{i \in [n]} x_i \ma_i$.
Given estimate of $\textup{OPT}$ exponentially bounded in $\tfrac{nd}{\eps}$, there is a procedure calling Algorithm~\ref{alg:boxpack} $O(\log\frac{nd}{\eps})$ times giving $x \in \Delta^n$ with $\norm{x}_\infty \le \tfrac{(1 + \alpha)(1 + \eps)}{n}$, $\norm{\alla(x)}_p \le (1 + \eps)\opt$. Algorithm~\ref{alg:boxpack} runs in $O(\tfrac{\log(nd/\eps)\log n}{\eps^2})$ iterations, each implementable in time $O(\nnz \cdot \frac{p\log(nd/\eps)}{\eps^2})$.
\end{restatable}

Our method, found in Appendix~\ref{app:pack}, approximately solves \eqref{eq:boxconstrainedschatten} by first applying a standard binary search to place $\alla(x)$ on the right scale, for which it suffices to solve an approximate decision problem. Then, we apply a truncated mirror descent procedure on the potential $\Phi(w) = \log(\exp(\norm{\alla(w)}_p) + \exp(\tfrac{n}{1 + \alpha}\norm{w}_\infty)) - \norm{w}_1$, and prove correctness for solving the decision problem following the framework we outlined in Section~\ref{ssec:mrwzonsimplex}. 	%

\section{Robust sub-Gaussian PCA in nearly-linear time}
\label{sec:subspace}

We give our nearly-linear time robust PCA method, leveraging developments of Section~\ref{sec:packing}. Throughout, we will be operating under Assumption~\ref{assume:corruption}, for some corruption parameter $\eps$ with $\eps\log\eps^{-1}\log d = O(1)$; $\eps = O(\frac{1}{\log d\log\log d})$ suffices. We now develop tools to prove Theorem~\ref{thm:main-robust}.

Algorithm~\ref{alg:robust-pca} uses three subroutines: our earlier $\OneD$ method (Lemma~\ref{lem:univariate}), an application of our earlier Proposition~\ref{prop:boxconstrainedp} to approximate the solution to
\begin{equation}\label{eq:robust-schatten}\min_{w \in \fS_\eps^n} \norm{\sum_{i \in [n]} w_i X_iX_i^\top}_p, \text{ for }p =\Theta\Par{\sqrt{\frac{\log d}{\eps\log \eps^{-1}}}},\end{equation}
and a method for computing approximate eigenvectors by \cite{MuscoM15} (discussed in Appendix~\ref{app:pca}).
\begin{restatable}{proposition}{poweriter}\label{prop:poweriter}
There is an algorithm $\mathsf{Power}$ (Algorithm 1, \cite{MuscoM15}), parameterized by $t \in [d]$, tolerance $\teps > 0$, $p \ge 1$, and $\ma \in \PSD^d$, which outputs orthonormal $\{z_j\}_{j \in [t]}$ with the guarantee
	\begin{equation}\label{eq:mpsvd}
	\begin{rcases}
	\left|z_j^\top \ma^p z_j - \lam_j^p (\ma) \right| &\le \teps \lam_j^p (\ma) \\
	\left|z_j^\top \ma^{p - 1} z_j - \lam_j^{p - 1} (\ma) \right|&\le \teps \lam_j^{p - 1} (\ma)
	\end{rcases} \text{ for all } j \in [t].
	\end{equation}
	Here, $\lam_j(\ma)$ is the $j^{th}$ largest eigenvalue of $\ma$. The total time required by the method is $O(\nnz(\ma) \tfrac{tp\log d}{\veps})$.
\end{restatable}

\begin{algorithm}
	\caption{$\RobustPCA(\{X_i\}_{i \in [n]}, \eps, t)$} 
	\begin{algorithmic}[1]\label{alg:robust-pca}
	\STATE \textbf{Input:} $\{X_i\}_{i \in [n]}$ $\eps = O(\frac{1}{\log d\log\log d})$, $t \in [d]$ with $\covar_{t + 1} \le (1 - \gamma)\covar$ for $\gamma$ in Theorem~\ref{thm:main-robust}
	\STATE $w \gets$ $\BoxPacking$ (Proposition~\ref{prop:boxconstrainedp}) on $\{\ma_i = X_iX_i^\top\}_{i \in [n]}$, $\alpha \gets \eps$, $p$ as in \eqref{eq:robust-schatten}
	\STATE $\mm = \sum_{i \in[n]} w_i X_i X_i^\top$
	\STATE $\{z_j\}_{j \in [t]} = \mathsf{Power}(t, \eps, p, \mm)$
	\STATE $\alpha_j \gets \mathsf{1DRobustVariance} (\{X_i\}_{i \in [n]}, \mm^{\frac{p-1}{2}}z_j/\|\mm^{\frac{p-1}{2}}z_j\|_2, \eps)$ for all $j \in [t]$
	\RETURN $z_{j^*}$ for $j^* = \argmax_{j \in [t]} \alpha_j$
	\end{algorithmic}	
\end{algorithm}
Algorithm~\ref{alg:robust-pca} is computationally bottlenecked by the application of Proposition~\ref{prop:boxconstrainedp} on Line 2 and the call to $\mathsf{Power}$ on Line 4, from which the runtime guarantee of Theorem~\ref{thm:main-robust} follows straightforwardly. To demonstrate correctness, we first certify the quality of the solution to \eqref{eq:robust-schatten}.
\begin{restatable}{lemma}{robustobjectiveub}
	\label{lem:robust-objective-ub}
	Let $n = \Omega\Par{\tfrac{d + \log \delta^{-1}}{(\eps \log \eps^{-1})^2}}$.
	With probability $1 - \tfrac{\delta}{2}$, the uniform distribution over $G$ attains value $(1 +\tfrac{\teps}{2}) \| \covar \|_p$ for objective \eqref{eq:robust-schatten}, where $\teps = C' \epsilon \log \eps^{-1}$ for a universal constant $C' > 0$.
\end{restatable}
The proof of this is similar to results in e.g.~\cite{DiakonikolasKKLMS19,li2018principled}, and combines concentration guarantees with a union bound over all possible corruption sets $B$. This implies the following immediately, upon applying the guarantees of Proposition~\ref{prop:boxconstrainedp}.
\begin{corollary}
	\label{cor:robust-output-ub}
	Let $w$ be the output of Line 2 of $\RobustPCA$. Then, we have
	$\norm{w}_\infty \le \frac{1}{(1 - 2\eps)n}$, and $\norm{\sum_{i \in [n]} w_i X_i X_i^\top}_p \le (1 + \teps)\norm{\covar}_p$ under the guarantee of Lemma~\ref{lem:robust-objective-ub}.
\end{corollary}

Let $w$ be the output of the solver.
Recall that $\mm = \sum_{i = 1}^n w_i X_i X_i^\top$.
Additionally, define
\begin{equation}\label{eq:mgbdef}\mmg \defeq \sum_{i \in G} w_i X_i X_i^\top,\; \wg \defeq \sum_{i \in G} w_i,\; \mmb \defeq \sum_{i \in B} w_i X_i X_i^\top,\; \wb \defeq \sum_{i \in G} w_i \; .\end{equation}
Notice in particular that $\mm = \mmg + \mmb$, and that all these matrices are PSD.
We next prove the second, crucial fact, which says that $\mmg$ is a good approximator to $\covar$ in Loewner ordering:
\begin{restatable}{lemma}{linfguarantee}\label{lem:linfguarantee}
	Let $n = \Omega\Par{\tfrac{d + \log \delta^{-1}}{(\eps \log \eps^{-1})^2}}$.
	With probability at least $1 - \tfrac{\delta}{2}$, 
	$(1 + \teps)\covar \succeq \mmg \succeq (1 - \teps) \covar$.
\end{restatable}
The proof combines the strategy in Lemma~\ref{lem:robust-objective-ub} with the guarantee of the SDP solver. Perhaps surprisingly, Corollary~\ref{cor:robust-output-ub} and Lemma~\ref{lem:linfguarantee} are the only two properties about $\mm$ that our final analysis of Theorem~\ref{thm:main-robust} will need. 
In particular, we have the following key geometric proposition, which carefully combines trace inequalities to argue that the corrupted points $\mmb$ cannot create too many new large eigendirections.
\begin{restatable}{proposition}{containsgamma}\label{prop:containsgamma}
	Let $\mm = \mmg + \mmb$ be so that $\norm{\mm}_p \leq (1 + \teps) \norm{\covar}_p$, $\mmg \succeq 0$ and $\mmb \succeq 0$, and so that $(1 + \teps)\covar \succeq \mmg \succeq (1 - \teps) \covar$.
	Following notation of Algorithm~\ref{alg:robust-pca}, let
	\begin{equation}\label{eq:eigdecomp}\mm = \sum_{j \in [d]} \lam_j v_j v_j^\top,\; \covar = \sum_{j \in [d]} \sigma_j u_j u_j^\top\end{equation}
	be sorted eigendecompositions of $\mm$ and $\covar$, so $\lambda_1 \geq \ldots \geq \lambda_d$, and $\sigma_1 \geq \ldots \geq \sigma_d$. Let $\gamma$ be as in Theorem~\ref{thm:main-robust}, and assume $\sigma_{t + 1} < (1 - \gamma)\sigma_1$. Then,
	\[\max_{j \in [t]} v_j^\top \covar v_j \ge (1 - \gamma)\norm{\covar}_\infty.\]
\end{restatable}
\begin{proof}
	For concreteness, we will define the parameters
	\[p = \frac{2}{7}\sqrt{\frac{\log(3d)}{\teps}},\; \gamma = 14\sqrt{\teps \log (3d)} = 49p\teps.\]
	For these choices of $p$, $\gamma$, we will use the following (loose) approximations for sufficiently small $\teps$:
	\begin{equation}\label{eq:taylorapproxes}\begin{aligned}
	\Par{1 - \frac{\gamma}{4}}^p = \Par{1 - \frac{\gamma}{4}}^{\frac{4}{\gamma}\log(3d)}\le \frac{1}{3d},\; (1 + \teps)^p - (1 - \teps)^p \le \exp(p\teps) - (1 - p\teps) \le 3p\teps.
	\end{aligned}\end{equation}
	Suppose for contradiction that all $v_j^\top \covar v_j < (1 - \gamma)\sigma_1$ for $j \in [t]$. By applying the guarantee of Corollary~\ref{cor:robust-output-ub} and Fact~\ref{fact:qnormcert}, it follows that
	\begin{equation}\label{eq:applysdp}\inprod{\mm}{\mm^{p - 1}} = \norm{\mm}_p^p \le (1 + \teps)^p \norm{\covar}_p^p.\end{equation}
	Let $s \in [d]$ be the largest index such that $\sigma_s > \Par{1 - \tfrac{\gamma}{4}}\sigma_1$, and note that $s \le t$. We define
	\[\mn \defeq \sum_{j \in [s]} \lam_j^{p - 1} v_j v_j^\top \preceq \mm^{p - 1}.\]
	That is, $\mn$ is the restriction of $\mm^{p - 1}$ to its top $s$ eigendirections. Then,
	\begin{equation}\label{eq:goodbaddirs}
	\begin{aligned}
	\inprod{\mm}{\mm^{p - 1}} &= \inprod{\mmb}{\mm^{p - 1}} + \inprod{\mmg}{\mm^{p - 1}} \\
	&\ge \inprod{\mmb}{\mm^{p - 1}} + \inprod{(1 - \teps)\covar}{\mm^{p - 1}} \ge \inprod{\mmb}{\mn} + (1 - \teps)^p\norm{\covar}_p^p.
	\end{aligned}
	\end{equation}
	In the second line, we used Lemma~\ref{lem:linfguarantee} twice, as well as the trace inequality Lemma~\ref{lem:pnormshift} with $\ma = \mm$ and $\mb = (1 - \teps)\covar$. Combining \eqref{eq:applysdp} with \eqref{eq:goodbaddirs}, and expanding the definition of $\mmb$, yields
	\begin{equation}\label{eq:lhsrhsexpand}
	\begin{aligned}
	\Par{(1 + \teps)^p - (1 - \teps)^p}\norm{\covar}_p^p &\ge \inprod{\mmb}{\mn} = \inprod{\mmb}{\sum_{j \in [s]}\lam_j^{p - 1} v_j v_j^\top}\\
	&= \inprod{\mm}{\sum_{j \in [s]} \lam_j^{p - 1} v_j v_j^\top} - \inprod{\mmg}{\sum_{j \in [s]} \lam_j^{p - 1} v_j v_j^\top} \\
	&\ge \inprod{\mm}{\sum_{j \in [s]} \lam_j^{p - 1} v_j v_j^\top} - (1 + \teps)\inprod{\covar}{\sum_{j \in [s]} \lam_j^{p - 1} v_j v_j^\top}\\
	&=\sum_{j \in [s]} \Par{\lam_j^p - (1 + \teps)\lam_j^{p - 1} v_j^\top \covar v_j} \ge \sum_{j \in [s]} \Par{\lam_j^p - \lam_j^{p - 1}(1 + \teps)(1 - \gamma)\sigma_1}.
	\end{aligned}
	\end{equation}
	The third line followed from from the spectral bound $\mmg \preceq (1 + \teps)\covar$ of Lemma~\ref{lem:linfguarantee}, and the fourth followed from the fact that $\{\lam_j\}_{j \in [d]}$, $\{v_j\}_{j \in [d]}$ eigendecompose $\mm$, as well as the assumption $v_j^\top \covar v_j \le (1 -\gamma)\sigma_1$ for all $j \in [t]$. Letting $S \defeq \sum_{j \in [s]} \sigma_j^p$, and using both approximations in \eqref{eq:taylorapproxes}, 
	\begin{equation}\label{eq:sdominates}\begin{aligned}\norm{\covar}_p^p  \le \sum_{j \in [s]} \sigma_j^p + \Par{1 - \frac{\gamma}{4}}^p (d - s)\sigma_1^p \le \frac{4}{3}S 
	\implies \Par{(1 + \teps)^p - (1 - \teps)^p}\norm{\covar}_p^p \le 4p\teps S.
	\end{aligned}\end{equation}
	Next, we bound the last term of \eqref{eq:lhsrhsexpand}. By using $(1 + \teps)(1 - \gamma) \le 1 - \tfrac{\gamma}{2}$, 
	\begin{equation}\label{eq:dealwithlhs}
	\begin{aligned}
	\sum_{j \in [s]} \Par{\lam_j^p- \lam_j^{p - 1}(1 + \teps)(1 - \gamma)\sigma_1} &\ge \sum_{j \in [s]} \lam_j^{p - 1}\Par{\lam_j - \Par{1 - \frac{\gamma}{2}} \sigma_1} \\
	&\ge \frac{\gamma}{6} \sum_{j \in [s]} \lam_j^{p - 1} \sigma_1 \ge \frac{\gamma}{6} \Par{1 - \teps}^{p - 1}\sum_{j \in [s]} \sigma_j^p  \ge \frac{\gamma}{12} S.
	\end{aligned}
	\end{equation}
	The second line used $\lam_j - (1 - \tfrac{\gamma}{2})\sigma_1 \ge (1 - \teps)\sigma_j - (1 - \tfrac{\gamma}{2})\sigma_1 \ge \tfrac{\gamma}{6}\sigma_1$ by definition of $s$,  Lemma~\ref{lem:minimaxchar} (twice), and $(1 - \teps)^{p - 1} \ge \thalf$. Combining \eqref{eq:dealwithlhs} and \eqref{eq:sdominates} and plugging into \eqref{eq:lhsrhsexpand},
	\[4p\teps S  \ge \frac{\gamma}{12} S \implies 48p\teps \ge \gamma.\]
	By the choice of $\gamma$ and $p$ (i.e.\ $\gamma = 49p\teps$), we attain a contradiction.
\end{proof}

In the proof of Proposition~\ref{prop:containsgamma}, we used the following facts.

\begin{lemma}\label{lem:pnormshift}
	Let $\ma \succeq \mb \succeq 0$ be symmetric matrices and $p$ a positive integer. Then we have
	\[\Tr\Par{\ma^{p - 1}\mb} \ge  \Tr\Par{\mb^p}. \]
\end{lemma}
\begin{proof}
	For any $1 \le k \le p - 1$,
	\[\Tr\Par{\ma^k \mb^{p - k}} \ge \Tr\Par{\ma^{k - 1}\mb^{\frac{p - k}{2}}\ma \mb^{\frac{p - k}{2}}} \ge \Tr\Par{\ma^{k - 1}\mb^{\frac{p - k}{2}}\mb\mb^{\frac{p - k}{2}}} =  \Tr\Par{\ma^{k - 1}\mb^{p - k + 1}}. \]
	The first step used the Extended Lieb-Thirring trace inequality $\Tr(\mm\mn^2) \ge \Tr(\mm^\alpha\mn\mm^{1 - \alpha}\mn)$ for $\alpha \in [0, 1]$, $\mm, \mn \in \PSD^d$ (see e.g.\ Lemma 2.1, \cite{Allen-ZhuLO16}), and the second $\ma \succeq \mb$. Finally, induction on $k$ yields the claim.
\end{proof}

\begin{lemma}\label{lem:minimaxchar}
	For all $j \in [d]$, $\lam_j \ge (1 - \teps) \sigma_j$.
\end{lemma}
\begin{proof}
	By the Courant-Fischer minimax characterization of eigenvalues,
	\[\lam_j \ge \min_{k \in [j]} u_k^\top \mm u_k.\]
	However, we also have $\mm \succeq \mmg \succeq (1 - \teps) \covar$ (Lemma~\ref{lem:linfguarantee}), yielding the conclusion.
\end{proof}

The guarantees of Proposition~\ref{prop:containsgamma} were geared towards exact eigenvectors of the matrix $\mm$. We now modify the analysis to tolerate inexactness in the eigenvector computation, in line with the processing of Line 5 of our Algorithm~\ref{alg:robust-pca}. This yields our final claim in Theorem~\ref{thm:main-robust}.

\begin{corollary}\label{corr:inexactgamma}
	In the setting of Proposition~\ref{prop:containsgamma}, and letting $\{z_j\}_{j \in [t]}$ satisfy \eqref{eq:mpsvd}, set for all $j \in [t]$
	\[y_j \defeq \frac{\mm^{\frac{p - 1}{2}} z_j}{\norm{\mm^{\frac{p - 1}{2}} z_j}_2}.\]
	Then with probability at least $1 - \delta$,
	\[\max_{j \in [t]} y_j^\top \covar y_j \ge (1 - \gamma) \norm{\covar}_\infty.\]
\end{corollary}
\begin{proof}
	Assume all $y_j$ have $y_j^\top \covar y_j \le (1 - \gamma)\sigma_1$ for contradiction. We outline modifications to the proof of Proposition~\ref{prop:containsgamma}. Specifically, we redefine the matrix $\mn$ by
	\[\mn \defeq \mm^{\frac{p - 1}{2}}\Par{\sum_{j \in [s]}z_jz_j^\top}\mm^{\frac{p - 1}{2}}.\]
	Because $\sum_{j \in [s]}z_jz_j^\top$ is a projection matrix, it is clear $\mn \preceq \mm^{p - 1}$. Therefore, by combining the derivations \eqref{eq:applysdp} and \eqref{eq:goodbaddirs}, it remains true that 
	\[\Par{(1 + \teps)^p - (1 - \teps)^p}\norm{\covar}_p^p \ge \inprod{\mmb}{\mn} = \inprod{\mm}{\mn} - \inprod{\mmg}{\mn}.\]
	We now bound these two terms in an analogous way from Proposition~\ref{prop:containsgamma}, with negligible loss; combining these bounds will again yield a contradiction. First, we have the lower bound
	\begin{align*}
	\inprod{\mm}{\sum_{j \in [s]} \mm^{\frac{p - 1}{2}} z_j z_j^\top \mm^{\frac{p - 1}{2}}} = \sum_{j \in [s]} z_j^\top \mm^p z_j \ge (1 - \teps) \sum_{j \in [s]} \lam_j^p.
	\end{align*}
	Here, the last inequality applied the assumption \eqref{eq:mpsvd} with respect to $\mm^p$. Next, we upper bound
	\begin{align*}
	\inprod{\mmg}{\sum_{j \in [s]} \mm^{\frac{p - 1}{2}} z_j z_j^\top \mm^{\frac{p - 1}{2}}} &\le (1 + \teps)\inprod{\covar}{\sum_{j \in [s]} \mm^{\frac{p - 1}{2}} z_j z_j^\top \mm^{\frac{p - 1}{2}}} \\
	&= (1 + \teps)\sum_{j \in [s]} \norm{\mm^{\frac{p - 1}{2}} z_j}_2^2 y_j^\top \covar y_j \\
	&\le (1 + \teps)(1 - \gamma)\sigma_1 \sum_{j \in [s]} z_j^\top \mm^{p - 1} z_j \\
	&\le (1 - \gamma)(1 + \teps)^2\sigma_1 \sum_{j \in [s]} \lam_j^{p - 1},
	\end{align*}
	The first line used $\mmg \preceq (1 + \teps)\covar$, the second used the definition of $y_j$, the third used our assumption $y_j^\top \covar y_j \le (1 - \gamma)\sigma_1$, and the last used \eqref{eq:mpsvd} with respect to $\mm^{p - 1}$. Finally, the remaining derivation \eqref{eq:dealwithlhs} is tolerant to additional factors of $1 + \teps$, yielding the same conclusion up to constants.
\end{proof}

Finally, we prove Theorem~\ref{thm:main-robust} by combining the tools developed thus far.
\begin{proof}[Proof of Theorem~\ref{thm:main-robust}]
Correctness of the algorithm is immediate from Corollary~\ref{corr:inexactgamma} and the guarantees of $\OneD$. Concretely, Corollary~\ref{corr:inexactgamma} guarantees that one of the vectors we produce will be a $(1 - \gamma)$-approximate top eigenvector (say some index $j \in [t]$), and $\OneD$ will only lose a negligible fraction $O(\eps\log\eps^{-1})$ of this quality (see Lemma~\ref{lem:univariate}); the best returned eigenvector as measured by $\OneD$ can only improve the guarantee. Finally, the failure probability follows by combining the guarantees of Lemmas~\ref{lem:univariate},~\ref{lem:robust-objective-ub}, and~\ref{lem:linfguarantee}.

We now discuss runtime. The complexity of lines 2, 4, and 5, as guaranteed by Proposition~\ref{prop:boxconstrainedp}, Proposition~\ref{prop:poweriter}, and Lemma~\ref{lem:univariate} are respectively (recalling $p = \tO(\eps^{-0.5})$)
\[\tO\Par{\frac{nd}{\eps^{4.5}}},\; \tO\Par{\frac{ndt}{\eps^{1.5}}},\; \tO\Par{ndt}.\]
Throughout we use that we can compute matrix-vector products in an arbitrary linear combination of the $X_i X_i^\top$ in time $O(nd)$; it is easy to check that in all runtime guarantees, $\nnz$ can be replaced by this computational cost. Combining these bounds yields the final conclusion.
\end{proof} 	
	\subsection*{Acknowledgments}
	
	We thank Swati Padmanabhan and Aaron Sidford for helpful discussions.
	
	\newpage
	\bibliographystyle{alpha}	
	\bibliography{schatten-packing}
	\newpage
	\begin{appendix}

\section{Concentration}
\label{app:concentration}

\subsection{Sub-Gaussian concentration}

We use the following concentration facts on sub-Gaussian distributions following from standard techniques, and give an application bounding Schatten-norm deviations. 
\begin{restatable}{lemma}{restatesubgconc}\label{lem:subg_conc}
	Under Assumption~\ref{assume:corruption}, there are universal constants $C_1$, $C_2$ such that
	\[\Pr\Brack{\sup_{\substack{v \in \R^d \\ \norm{v}_2 = 1}} \left|v^\top \Par{\frac{1}{n}\sum_{i \in G'} X_i X_i^\top - \covar}v\right| - tv^\top\covar v > 0} \le \exp\Par{C_1 d - C_2 n\min(t, t^2)}.\]
\end{restatable}

\begin{proof}
	By observing \eqref{eq:subg_def}, it is clear that the random vector $\tX = \covar^{-\half} X$ for $X \sim \dist$ has covariance $\id$ and sub-Gaussian proxy $c\id$. For any fixed unit vector $u$, by Lemma 1.12 of \cite{RigolletH17}, the random variable $(u^\top \tX)^2 - 1$ is sub-exponential with parameter $16c$, so by Bernstein's inequality (Theorem 1.13, \cite{RigolletH17}), defining $\tX_i = \covar^{-\half} X_i$ for each $X_i \sim \dist$,
	\[\Pr\Brack{\left|u^\top \Par{\frac{1}{n}\sum_{i \in G'} \tX_i \tX_i^\top - \id}u\right| > \frac{t}{2}} \le \exp\Par{-\frac{n}{2^{11}c^2}\min(t, t^2)}.\]
	For shorthand define $\mm \defeq \tfrac{1}{n}\sum_{i \in G'} \tX_i \tX_i^\top$, and let $\net$ be a maximal $\tfrac{1}{4}$-net of the unit ball (as measured in $\ell_2$ distance). By Lemma 1.18 of \cite{RigolletH17}, $|\net| \le 12^d$, so by a union bound,
	\[\Pr\Brack{\sup_{u \in \mathcal{N}}\left|u^\top (\mm - \id) u\right| > \frac{t}{2}} \le \exp\Par{3d -\frac{n}{2^{11}c^2}\min(t, t^2)}.\]
	Next, by a standard application of the triangle inequality (see e.g.\ Exercise 4.3.3, \cite{Vershynin16})
	\[
	\sup_{\substack{v \in \R^d \\ \norm{v}_2 = 1}} \left|v^\top (\mm - \id)v\right| \le 2\sup_{u \in \net} \left|u^\top (\mm - \id) u\right| \le t
	\]
	with probability at least $1 - \exp\Par{C_1 d - C_2 n\min(t, t^2)}$ for appropriate $C_1$, $C_2$. The conclusion follows since its statement is scale invariant, so it suffices to show as we have that
	\[\Pr\Brack{\sup_{\substack{v \in \R^d \\ \norm{v}_{\covar} = 1}} \left|v^\top \Par{\frac{1}{n}\sum_{i \in G'} X_i X_i^\top - \covar}v\right| - tv^\top\covar v > 0} \le \exp\Par{C_1 d - C_2 n\min(t, t^2)}.\]
\end{proof}

\begin{restatable}{corollary}{restatelpconc}\label{corr:lpconc}
	Let $p \ge 2$. Under Assumption~\ref{assume:corruption}, there are universal constants $C_1$, $C_2$ with
	\[\Pr\Brack{\norm{\frac{1}{n}\sum_{i \in G'} X_i X_i^\top - \covar}_p > t\norm{\covar}_p} \le \exp\Par{C_1 d - C_2 n\min(t, t^2)}.\]
\end{restatable}
\begin{proof}
	Suppose the event in Lemma~\ref{lem:subg_conc} does not hold, which happens with probability at least $1 - \exp(C_1 d - C_2 n \min(t, t^2))$. Define for shorthand $\mm \defeq \tfrac{1}{n}\sum_{i \in G'} X_i X_i^\top - \covar$ and let its spectral decomposition be $\sum_{j \in [d]} \lam_j v_j v_j^\top$. By the triangle inequality and Fact~\ref{fact:qnormcert},
	\begin{align*}
	\norm{\mm}_p &\le \sum_{j \in [d]} \frac{|\lam_j|^{p - 1}}{\norm{\mm}_p^{p - 1}} \left|v_j^\top \Par{\frac{1}{n}\sum_{i \in G'} X_i X_i^\top - \covar} v_j\right| \\
	&\le t\sum_{j \in [d]} \frac{|\lam_j|^{p - 1}}{\norm{\mm}_p^{p - 1}} v_j^\top \covar v_j = t\inprod{\sum_{j \in [d]} \frac{|\lam_j|^{p - 1}}{\norm{\mm}_p^{p - 1}} v_j v_j^\top}{\covar} \le t\norm{\covar}_p.
	\end{align*}
	In the last inequality, we used that $\sum_{j \in [d]} \tfrac{|\lam_j|^{p - 1}}{\norm{\mm}_p^{p - 1}} v_j v_j^\top$ has unit $\ell_q$ norm, and applied Fact~\ref{fact:qnormcert}.
\end{proof}

\subsection{Concentration under weightings in $\fS_\eps^n$}

We consider concentration of the empirical covariance under weightings which are not far from uniform, in spectral and Schatten senses.

	\begin{lemma}\label{lem:pnormbound_fs}
	Under Assumption~\ref{assume:corruption}, let $\delta \in [0, 1]$, $p \ge 2$, and $n = \Omega\Par{\tfrac{d + \log \delta^{-1}}{(\eps \log \eps^{-1})^2}}$ for a sufficiently large constant. Then for a universal constant $C_3$,
	\[\Pr\Brack{\exists w \in \fS_\eps^n \; \Bigg|\; \norm{\sum_{i \in G'} w_i X_i X_i^\top - \covar}_p > C_3 \cdot \eps\log\eps^{-1} \norm{\covar}_p} \le \frac{\delta}{2}.\]
\end{lemma}
\begin{proof}
	Because the vertices of $\fS_\eps^n$ are uniform over sets $S \subseteq G'$ with $|S| = (1 - \eps)n$ (see e.g.\ Section 4.1, \cite{DiakonikolasKKLMS19}), by convexity of the Schatten-$p$ norm it suffices to prove
	\[\Pr\Brack{\exists S \text{ with } |S| = (1 - \eps)n \; \Bigg|\; \norm{\frac{1}{(1 - \eps)n}\sum_{i \in S} X_i X_i^\top - \covar}_p > C_3 \cdot \eps\log\eps^{-1} \norm{\covar}_p} \le \frac{\delta}{4}.\]
	For any fixed $S$, and recalling $|S^c| = \eps n$, we can decompose this sum as
	\begin{equation}\label{eq:good_all_bad}\frac{1}{(1 - \eps)n}\sum_{i \in S} X_i X_i^\top = \frac{1}{1 - \eps} \Par{\frac{1}{n}\sum_{i \in G'} X_i X_i^\top} - \frac{\eps}{1 - \eps}\Par{\frac{1}{|S^c|}\sum_{i \in S^c}X_i X_i^\top}.\end{equation}
	By applying Corollary~\ref{corr:lpconc}, it follows that by setting $t = \tfrac{1 - \eps}{2} \cdot \eps\log\eps^{-1}$ and our choice of $n$ that 
	\begin{equation}\label{eq:bigbound}\Pr\Brack{\norm{\frac{1}{n}\sum_{i \in G'} X_i X_i^\top - \covar}_p > \frac{1 - \eps}{2} \cdot \eps\log\eps^{-1}\norm{\covar}_p} \le \frac{\delta}{4}.\end{equation}
	Moreover, for any fixed $S^c$, setting $t = \tfrac{1 - \eps}{2} \cdot C_3 \log\eps^{-1}$ where $C_3$ is a sufficiently large constant, so that for sufficiently small $\eps$, $t = \min(t, t^2)$,
	\begin{equation}\label{eq:eachscbound}\begin{aligned}\Pr\Brack{\norm{\frac{1}{\eps n}\sum_{i \in S^c} X_i X_i^\top - \covar}_p > \frac{1 - \eps}{2} \cdot C_3 \cdot \log\eps^{-1}\norm{\covar}_p} &\le \exp\Par{C_1 d - C_2 \eps n t} \\
	&\le \exp\Par{-\Par{\log \delta^{-1} + n\eps\log\eps^{-1}}} \\
	&\le \frac{\delta}{4\binom{n}{\eps n}}.\end{aligned}\end{equation}
	Here, we used that $\log \binom{n}{\eps n} = O\Par{n\eps \log \eps^{-1}}$. Finally, union bounding over all possible sets $S^c$ imply that with probability at least $1 - \tfrac{\delta}{2}$, the following events hold:
	\begin{align*}\norm{\frac{1}{n}\sum_{i \in G'} X_i X_i^\top - \covar}_p < \frac{1 - \eps}{2} \cdot \eps\log\eps^{-1}\norm{\covar}_p,\\ \norm{\frac{1}{|S^c|}\sum_{i \in S^c} X_i X_i^\top - \covar}_p < \frac{1 - \eps}{2} \cdot C_3 \cdot \log\eps^{-1}\norm{\covar}_p \text{ for all } S \text{ with } |S| = (1 - \eps) n.\end{align*}
	Combining these bounds in the context of \eqref{eq:good_all_bad} after applying the triangle inequality, we have with probability at least $1 - \tfrac{\delta}{2}$ for all $S$ the desired conclusion,
	\[\norm{\frac{1}{(1 - \eps)n}\sum_{i \in S} X_i X_i^\top - \covar}_p < C_3 \cdot \eps\log\eps^{-1} \norm{\covar}_p.\]
\end{proof}

\begin{corollary}\label{corr:infnormbound_fs}
Under Assumption~\ref{assume:corruption}, let $n = \Omega\Par{\tfrac{d + \log\delta^{-1}}{(\eps\log\eps^{-1})^2}}$ for a sufficiently large constant. For universal $C_3$ and all $w \in \fS_\eps^n$, with probability at least $1 - \tfrac{\delta}{2}$,
	\[C_3 \cdot \eps\log\eps^{-1} \covar\succeq \sum_{i \in G'} w_i X_i X_i^\top - \covar \succeq - C_3 \cdot \eps\log\eps^{-1} \covar.\]
\end{corollary}
\begin{proof}
	Consider any unit vector $v \in \R^d$. By similar arguments as in \eqref{eq:bigbound}, \eqref{eq:eachscbound}, and applying a union bound over all $S$ with $|S| = (1 - \eps)n$, with probability at least $1 - \tfrac{\delta}{2}$, it follows from Lemma~\ref{lem:subg_conc} that
	\begin{align}\left|v^\top \Par{\frac{1}{n}\sum_{i \in G'} X_i X_i^\top - \covar} v\right| < \frac{1 - \eps}{2} \cdot \eps\log\eps^{-1}v^\top \covar v, \label{eq:conc-per-vector-1} \\
	\left|v^\top\Par{\frac{1}{|S^c|}\sum_{i \in S^c}X_i X_i^\top - \covar}v\right| < \frac{1 - \eps}{2} \cdot C_3 \cdot \log\eps^{-1} v^\top\covar v \; . \label{eq:conc-per-vector-2}\end{align}
	Therefore, again using the formula \eqref{eq:good_all_bad} and the triangle inequality yields the desired conclusion for all directions $v$, which is equivalent to the spectral bound of the lemma statement.
\end{proof} 	%

\section{Deferred proofs from Section~\ref{sec:filtering}}
\label{app:filtering}

\subsection{Robust univariate variance estimation}

In this section, we prove Lemma~\ref{lem:univariate}, which allows us to robustly estimate the quadratic form of a vector in the covariance of a sub-Gaussian distribution from corrupted samples. Algorithm~\ref{alg:univariate} is folklore, and intuitively very simple; it projects all samples onto $u$, throws away the $2\epsilon$ fraction of points with largest magnitude in this direction, and takes the mean of the remaining set. 

\begin{algorithm}\caption{Univariate variance estimation: $\OneD(\{X_i\}_{i \in [n]}, \eps, u)$} 
	\begin{algorithmic}\label{alg:univariate}
		\STATE \textbf{Input:} $\{X_i\}_{i \in [n]}$, $\eps > 0$, and a unit vector $u$
		\STATE Let $a_i = \left\langle X_i, u \right\rangle^2$ for $i = 1, \ldots, n$ 
		\STATE Sort the $a_i$ in increasing order. WLOG assume $a_1 \leq a_2 \leq \ldots \leq a_n$.
		\RETURN $\sigma_u^2 = \frac{1}{(1 - 2 \epsilon) n} \sum_{i = 1}^{(1 - 2 \eps) n} a_i$
	\end{algorithmic}
\end{algorithm}

We require the following helper fact.
\begin{fact}
	\label{fact:subexp-conditioning}
	Let $Z$ be a sub-exponential random variable with parameter at most $\lambda$\footnote{We say mean-zero $Z$ is sub-exponential with parameter $\lam$ if $\forall |s| \le \lam^{-1}$, $\E[\exp(sZ)] \le \exp(\tfrac{s^2\lam^2}{2})$.}, and let $\eps \in [0, 1]$. Then, for any event $E$ with $\Pr [Z \in E] \leq \epsilon$, 
	$| \E \left[Z \cdot \one[Z \in E] \right]|  \leq 2 \lambda \eps \log\eps^{-1}$.
\end{fact}
\begin{proof} 
We have by H\"older's inequality that for any $p, q \ge 1$ with $p^{-1} + q^{-1} = 1$,
\[\left| \E \left[ Z \cdot \one[Z \in E] \right] \right| \leq \E [|Z|^p]^{1/p} \cdot \epsilon^{1/q} \leq 2 \lambda p \cdot \epsilon^{1/q}.\]
The second inequality is Lemma 1.10~\cite{RigolletH17}. Setting $p = \log \eps^{-1}$ yields the result.
\end{proof}
\univariate*
\begin{proof}
The runtime claim is immediate; we now turn our attention to correctness. We follow notation of Assumption~\ref{assume:corruption}, and in a slight abuse of notation, also define $a_i = \inprod{X_i}{u}^2$ for $i \in G'$. First, for $X \sim \dist$, then $\inprod{u}{X}^2 - u^\top \covar u$ is sub-exponential with parameter at most $16c u^\top \covar u$ (Lemma 1.12, \cite{RigolletH17}). 	By Bernstein's inequality, we have that if $X \sim \dist$, then for all $t \geq 1$,
\begin{equation}
\label{eq:subexp-def}
\Pr \left[ \inprod{X}{u}^2 > 32c t u^\top \covar u \right] \leq \exp (-t) \; .
\end{equation}
Using this in a standard Chernoff bound, we have that with probability $1 - \tfrac{\delta}{2}$,
\begin{equation}
\label{eq:subexp-condition}
\frac{\left| \{i \in G': a_i > 64 c \log \eps^{-1} \cdot u^\top \covar u \} \right|}{n} \leq \epsilon \; .	 	
\end{equation}
Let $T = 64 c \log \eps^{-1} \cdot u^\top \covar u$, and let $Y$ be distributed as $(\inprod{u}{X}^2 - u^\top \covar u) \cdot \one[\inprod{u}{X}^2 \leq T]$, where $X \sim \dist$.	We observe $Y - \E[Y]$ is also sub-exponential with parameter $16cu^\top\covar u$, and that by Fact~\ref{fact:subexp-conditioning},
\begin{equation}\label{eq:yexpectbound}|\E[Y]| \le 32cu^\top\covar u\eps\log\eps^{-1}.\end{equation}
Define the interval $I = [0, T]$ and let $S$ be the set of points in $[n]$ that survive the truncation procedure, so that $\sigma^2_u = \tfrac{1}{|S|} \sum_{i \in S} a_i$. Given event~\eqref{eq:subexp-condition}, $a_i \in I$ for all $i \in S$, since there are at most $\eps n$ points in $G$ outside $I$, and $|B| \leq \eps n$.	We decompose the deviation as follows:
\begin{equation}\label{eq:totaldeviation}
	\begin{aligned}
	\sum_{i \in S} a_i - |S| u^\top \covar u &= \sum_{i \in G \cap S} (a_i - u^\top \covar u) + \sum_{i \in B \cap S} (a_i - u^\top \covar u) \\
	&= \sum_{i \in G' \cap I} (a_i - u^\top \covar u) + \sum_{i \in B \cap S} (a_i - u^\top \covar u) \\
	&- \sum_{i \in (G' \setminus G) \cap I} (a_i - u^\top \covar u) - \sum_{i \in (G \cap I) \setminus S} (a_i - u^\top \covar u).
	\end{aligned}
\end{equation}
Here we overloaded $i \in I$ to mean that $a_i$ lies in the interval $I$, and conditioned on $S$ lying entirely in $I$.
	We bound each of these terms individually.
	First, for all $i \in G' \cap I$, conditioning on \eqref{eq:subexp-condition} (i.e.\ all $a_i \in I$),  $a_i- u^\top \covar u$ is an independent sample from $Y$. Thus, by \eqref{eq:yexpectbound} and Bernstein's inequality,
\begin{equation}\label{eq:subexp-term-1}
	\begin{aligned}
	\left| \frac{1}{|G' \cap I|} \sum_{i \in G' \cap I} (a_i - u^\top \covar u) \right| &\leq \left| \frac{1}{|G' \cap I|} \sum_{i \in G' \cap I} (a_i - u^\top \covar u) - \E[Y] \right| + 32c u^\top \covar u \eps \log\eps^{-1} \\
	&\leq 64 c \cdot u^\top \covar u \epsilon \log \eps^{-1},
	\end{aligned}
\end{equation}
with (conditional) probability at least $1 - \tfrac{\delta}{2}$. By a union bound, both events occur with probability at least $1 - \delta$; condition on this for the remainder of the proof. Under this assumption, we control the other three terms of \eqref{eq:totaldeviation}. Observe that $|B \cap S| \leq \eps n$, $|(G' \setminus G) \cap I| \leq \eps n$, and $|(G\cap I) \setminus S| \leq \eps n$. Further, by definition of $I$, every summand is at most $64c\log\eps^{-1} \cdot u^\top\covar u$. Thus,
	\begin{align}
	\left| \sum_{i \in B \cap S} (a_i - u^\top \covar u) \right| &\leq 64c\eps n\log\eps^{-1}\cdot u^\top\covar u, \label{eq:subexp-term-2}\\
	\left| \sum_{i \in (G' \setminus G) \cap I} (a_i - u^\top \covar u) \right| &\leq 64c\eps n\log\eps^{-1}\cdot u^\top\covar u, \label{eq:subexp-term-3}\\
		\left| \sum_{i \in (G' \cap I) \setminus S} (a_i - u^\top \covar u) \right| &\leq 64c\eps n\log\eps^{-1}\cdot u^\top\covar u. \label{eq:subexp-term-4}
	\end{align}
	Combining~\eqref{eq:subexp-term-1},~\eqref{eq:subexp-term-2},~\eqref{eq:subexp-term-3}, and~\eqref{eq:subexp-term-4} in derivation \eqref{eq:totaldeviation} and dividing by $|S|$ yields the claim.
\end{proof}

Finally, we also give an alternative set of conditions under which we can certify correctness of $\OneD$. Specifically, this assumption will be useful in lifting indpendence assumptions between $u$ and our samples $\{X_i\}_{i \in [n]}$ in repeated calls within Algorithm~\ref{alg:filter}.

\begin{assumption}\label{assume:poly-cond}
Under Assumption~\ref{assume:corruption}, let the following conditions hold for universal constant $C_4$:
\begin{align}
C_4 \epsilon \log \eps^{-1} \cdot \covar &\succeq \frac{1}{n} \sum_{i \in G'} X_i X_i^\top - \covar \succeq -C_4 \epsilon \log \eps^{-1} \cdot \covar,
\label{eq:poly-cond-2} \\
C_4 \log \eps^{-1} \cdot \covar &\succeq \sum_{i \in G'} w_i \Par{X_i X_i^\top - \covar}\succeq -C_4 \log \eps^{-1} \cdot \covar \; \mbox{for all $w \in \fS_{1 - \eps}^n$}.\label{eq:poly-cond-3}
\end{align}
\end{assumption}
Note that \eqref{eq:poly-cond-3} is a factor $\eps$ weaker in its guarantee than Corollary~\ref{corr:infnormbound_fs}, and is over weights in a different set $\fS_{1 - \eps}^n$. Standard sub-Gaussian concentration (i.e.\ an unweighted variant of Corollary~\ref{corr:infnormbound_fs}) and modifying the proof of Corollary~\ref{corr:infnormbound_fs} to take the constraint set $\fS_{1 - \eps}^n$ and normalizing over vertex sets of size $\eps n$ yield the following conclusion.
\begin{lemma}\label{lem:poly-cond-holds}
Let $n = \Omega\Par{\tfrac{d + \log\delta^{-1}}{(\eps\log\eps^{-1})^2}}$ for a sufficiently large constant. Assumption~\ref{assume:poly-cond} holds with probability at least $1 - \tfrac{\delta}{2}$.
\end{lemma}

We give a variant of Lemma~\ref{lem:univariate} with slightly stronger guarantees for $\OneD$; specifically, it holds for all $u$ simultaneously for a fixed set of samples satisfying Assumption~\ref{assume:poly-cond}.

\begin{corollary}\label{corr:oned-independence}
Under Assumption~\ref{assume:poly-cond}, Algorithm~\ref{alg:univariate} outputs $\sigma_u^2$ with $|u^\top \covar u - \sigma_u^2| < Cu^\top\covar u \cdot \eps\log\eps^{-1}$, for $C$ a fixed multiple of the parameter $c$ in Assumption~\ref{assume:corruption}, and runs in time $O(nd + n\log n)$.
\end{corollary}
\begin{proof}
We discuss how to modify the derivations from Lemma~\ref{lem:univariate} appropriately in the absence of applications of Bernstein's inequality. First, note that appropriately combining \eqref{eq:poly-cond-2} and \eqref{eq:poly-cond-3} in a derivation such as \eqref{eq:good_all_bad} yields the following bound (deterministically under Assumption~\ref{assume:poly-cond}):
\begin{equation}\label{eq:poly-cond-4}
C_4 \epsilon \log \eps^{-1} \cdot \covar \succeq \sum_{i \in G'} w_i \Par{X_i X_i^\top - \covar}\succeq -C_4 \epsilon \log \eps^{-1} \covar \; \mbox{for all $w \in \fS_{3\eps}^n$} .
\end{equation}
Now, consider the decomposition \eqref{eq:totaldeviation}. We claim first that similarly to \eqref{eq:subexp-term-2}, \eqref{eq:subexp-term-3}, \eqref{eq:subexp-term-4} we can bound each summand in the latter three terms by $O(u^\top\covar u \log\eps^{-1})$; to prove this, it suffices to show that at least one filtered $a_i$ attains this bound, as then by definition of the algorithm, each non-filtered $a_i$ will as well. Note that a fraction between $\eps$ and $2\eps$ of points in $G \subset G'$ is filtered (since there are only $\eps n$ points from $B$). The assumption \eqref{eq:poly-cond-3} then implies precisely the desired bound on some filtered $a_i$ by placing uniform mass on filtered points from $G$, and applying pigeonhole. So, all non-filtered $a_i$ are bounded by $O(u^\top \covar u \log\eps^{-1})$, yielding analogous statements to \eqref{eq:subexp-term-2}, \eqref{eq:subexp-term-3}, \eqref{eq:subexp-term-4}.

Finally, an analogous derivation to \eqref{eq:subexp-term-1} follows via an application of the bound \eqref{eq:poly-cond-4}, where we place uniform mass on the set $G' \cap I$ and adjust constants appropriately, since the above argument shows that under the assumption \eqref{eq:poly-cond-3}, we have that at most $2\eps n$ indices $i \in G'$ have $a_i \not\in I$.
\end{proof}

\subsection{Preliminaries}

For convenience, we give the following preliminaries before embarking on our proof of Theorem~\ref{thm:poly-final} and giving guarantees on Algorithm~\ref{alg:filter}. First, we state a set of assumptions which augments Assumption~\ref{assume:poly-cond} with one additional condition, used in bounding the iteration count of our algorithm.

\begin{assumption}\label{assume:total-filter}
Under Assumption~\ref{assume:corruption}, let Assumption~\ref{assume:poly-cond} hold, as well as the following additional condition for the same universal constant $C_4$:
\begin{equation}
\norm{X_i}_2^2 \leq C_4\log\frac{n}{\delta} \cdot \Tr(\covar) \; \mbox{for all $i \in G$}. \label{eq:poly-cond-1}
\end{equation}
\end{assumption}

Standard sub-Gaussian concentration inequalities and a union bound, combined with our earlier claim Lemma~\ref{lem:poly-cond-holds}, then yield the following guarantee.
\begin{lemma}\label{lem:total-filter-holds}
Let $n = \Omega\Par{\tfrac{d + \log\delta^{-1}}{(\eps\log\eps^{-1})^2}}$ for a sufficiently large constant. Assumption~\ref{assume:total-filter} holds with probability at least $1 - \delta$.
\end{lemma}

\subsection{Analysis of $\PCAFilter$}
For this section, for any nonnegative weights $w$, define $\mm(w) \defeq \sum_{i \in [n]} w_i X_i X_i^\top$. We now state our algorithm, $\PCAFilter$. At all iterations $t$, it maintains a current nonnegative weight vector $w^{(t)}$ (initialized to be the uniform distribution on $[n]$), preserving the following invariants for all $t$:
\begin{align}
w_i^{(t - 1)} \ge w_i^{(t)} \text{ for all } i \in [n],\; \sum_{i \in B} w^{(t - 1)}_i - w^{(t)}_i \geq \sum_{i \in G} w^{(t - 1)}_i - w^{(t)}_i.\label{eq:poly-inv-2}
\end{align}

We now state our method as Algorithm~\ref{alg:filter}; note that the update to $w^{(t)}$ is of the form in Lemma~\ref{lem:uni-filter}.

\begin{algorithm}[H]
	\caption{$\PCAFilter(\{X_i\}_{i \in [n]}, \eps)$}
	\begin{algorithmic}[1]\label{alg:filter}
		\STATE Remove all points $i \in [n]$ with $\norm{X_i}_2^2 > c \log(\tfrac{n}{\delta}) \cdot \Tr (\covar) $
		\STATE $w^{(0)}_i \gets \tfrac{1}{n}$ for all $i \in [n]$, $t \gets 1$
		\STATE $u_1 \gets$ approximate top eigenvector of $\mm(w^{(0)})$
		\STATE $\sigma_1^2 \gets \OneD(\{X_i\}_{i \in [n]}, \eps, u_1)$
		\WHILE{$u_t^\top \mm(w^{(t - 1)}) u_t > (1 + 5C_5 \eps \log \eps^{-1}) \sigma_t^2$, where $C_5 = \max(C, C_4)$ from constants in Assumption~\ref{assume:poly-cond}, Corollary~\ref{corr:oned-independence}}
		\STATE $a_i \gets \inprod{u_t}{X_i}^2$ for all $i \in [n]$
		\STATE Sort (permute) the indices $[n]$ so the $a_i$ are in increasing order (with $a_1$ smallest, $a_n$ largest)
		\STATE Let $\ell$ be the largest index with $\sum_{i = \ell}^n w_i \geq 2 \epsilon$
		\STATE Define
		\[
		w_i^{(t)} \gets \begin{cases} \Par{1 - \frac{a_i}{a_n}} w_i^{(t - 1)}& \ell \le i \le n \\ w_i^{(t - 1)} & i < \ell\end{cases}
		\]
		\STATE $u_{t} \gets $ approximate top eigenvector of $\mm(w^{(t)})$
		\STATE $\sigma_{t}^2 \gets \OneD(\{X_i\}_{i \in [n]}, u_{t}, \eps)$
		\STATE $t \gets t + 1$
		\ENDWHILE
		\RETURN{$u_t$}
	\end{algorithmic}
\end{algorithm}

We assume that in Line 8, we also have $\sum_{i = \ell}^n w_i \le 3\eps$, as we can assume at least one point is corrupted i.e.\ $\eps \ge \tfrac{1}{n}$ (else standard algorithms suffice for our setting), so adding an additional $w_i$ can only change the sum by $\eps$. We first prove invariants \eqref{eq:poly-inv-2} are preserved; at a high level, we simply demonstrate that Lemma~\ref{lem:uni-filter} holds via concentration on $G$ and lack of termination.

\begin{lemma}
	\label{lem:inv-holds}
	Under Assumption~\ref{assume:poly-cond}, for any iteration $t$ of Algorithm~\ref{alg:filter}, suppose~\eqref{eq:poly-inv-2} held for all iterations $t' \leq t - 1$. Then, \eqref{eq:poly-inv-2} holds at iteration $t$.
\end{lemma}
\begin{proof}
	The first part of \eqref{eq:poly-inv-2} is immediate by observing the update in Line 9, so we show the second. We drop subscripts and superscripts for conciseness and focus on a single iteration $t$. Let $I_B = \{\ell, \ldots, n\} \cap B$, and $I_G = \{\ell, \ldots, n\} \cap G$. By Lemma~\ref{lem:uni-filter}, it suffices to demonstrate that 
	\begin{equation}\label{eq:suffice_invariant}\sum_{i \in I_B} w_i a_i > \sum_{i \in I_G} w_i a_i.\end{equation}
	First, $\sum_{i \in I_B} w_i \leq \epsilon$, so by definition of index $\ell$, we have $\eps \le \sum_{i \in I_G} w_i \le 2\epsilon$. Define $\tw_i = \tfrac{w_i}{\sum_{i \in I_G} w_i}$ if $i \in I_G$, and $0$ otherwise, and observe $\widetilde{w} \in \fS_{1 - 2\eps}^n$.
	By modifying constants appropriately from~\eqref{eq:poly-cond-3}, it follows from definition of $a_i = u^\top X_iX_i^\top u$ that
	\begin{equation}
	\label{eq:ig}
	\sum_{i \in I_G} w_i a_i \leq \Par{\sum_{i \in I_G} w_i} \cdot C_4\log\eps^{-1} \cdot u^\top \covar u \leq 2C_4 \eps \log\eps^{-1} \cdot u^\top \covar u.
	\end{equation}
	On the other hand, by~\eqref{eq:poly-cond-4} we know that the total quadratic form over $G$ is bounded as
	\begin{equation}\label{eq:goodwa}
	\sum_{i \in G} w_i a_i < \Par{\sum_{i \in G} w_i} \Par{1 + C_4\eps\log\eps^{-1}} u^\top \covar u < \Par{1 + C_4\eps\log\eps^{-1}} u^\top \covar u.
	\end{equation}
	Here, we applied the observation that the normalized $w_i$ restricted to $G$ are in $\fS^n_{1 - 3\eps}$ (e.g.\ using Lemma~\ref{lem:vuniform} inductively).
	However, since we did not terminate (Line 5), we must have by $u_t$ being a top eigenvector and Corollary~\ref{corr:oned-independence} (we defer discussions of inexactness to Theorem~\ref{thm:poly-final}) that
	\begin{align*}
	\sum_{i \in [n]} w_i a_i \ge (1 + 5C_5\eps\log\eps^{-1}) \sigma_t^2 \ge (1 + 4C_4\eps\log\eps^{-1}) \cdot u^\top \covar u\\ \implies \sum_{i \in B} w_i a_i > 3C_4\eps \log\eps^{-1} \cdot u^\top \covar u.
	\end{align*}
	To obtain the last conclusion, we used \eqref{eq:goodwa}. Finally, note that for all $i \in B \setminus I_B$,
	\[
	a_i \le a_\ell \leq \sum_{i \in I_G} \widetilde{w}_i a_i \leq C_4\log\eps^{-1} \cdot u^\top \covar u
	\]
	by rearranging~\eqref{eq:ig}. This implies that
	\[
	\sum_{i \in B \setminus I_B} w_i a_i \leq \Par{\sum_{i \in B \setminus I_B} w_i} \cdot C_4\log\eps^{-1} \cdot u^\top \covar u \leq C_4\epsilon \log \eps^{-1} \cdot u^\top \covar u.
	\]
	Thus, the desired inequality \eqref{eq:suffice_invariant} follows from combining the above derivations, e.g.\ using \eqref{eq:ig} and
	\[
	\sum_{i \in I_B} w_i a_i = \sum_{i \in B} w_i a_i - \sum_{i \in B \setminus I_B} w_i a_i > 2C_4\eps\log\eps^{-1} \cdot u^\top \covar u.
	\]
\end{proof} 
Lemma~\ref{lem:inv-holds} yields for all $t$ that $\sum_{i \in B} w^{(0)}_i - w^{(t)}_i \geq \sum_{i \in G} w^{(0)}_i - w^{(t)}_i$ by telescoping. Note that we can only remove at most $2 \eps$ mass from $w$ total, as $\sum_{i \in B} w^{(0)}_i - w^{(t)}_i \leq \epsilon$. Denote for shorthand normalized weights $v^{(t)} \defeq \tfrac{w^{(t)}}{\norm{w^{(t)}}_1}$. Then, the following is immediate by $\norm{w^{(t)}}_1 \ge 1 - 2\eps$.
\begin{lemma}\label{lem:vuniform}
	Under Assumption~\ref{assume:poly-cond}, in all iterations $t$ of Algorithm~\ref{alg:filter}, $v^{(t)} \in \fS_{2\eps}^n$.
\end{lemma}
Using Lemma~\ref{lem:vuniform}, we show that the output has the desired quality of being a large eigenvector.

\begin{lemma}
	\label{lem:output-correct}
	Under Assumption~\ref{assume:poly-cond}, let the output of Algorithm~\ref{alg:filter} be $u_T$.  Then for a universal constant $C^\star$, $u_T^\top \covar u_T \geq (1 - C^\star \eps \log \eps^{-1}) \| \covar \|_\infty$.
\end{lemma}
\begin{proof}
We assume for now that $u_T$ is an exact top eigenvector, and discuss inexactness while proving Theorem~\ref{thm:poly-final}. By~\eqref{eq:poly-cond-4} and Lemma~\ref{lem:vuniform}, as then the normalized restriction of $w^{(T)}$ to $G$ is in $\fS_{3\eps}^n$, 
\begin{align*}
\mm(w^{(T)}) \succeq \sum_{i \in G} w^{(T)}_i X_i X_i^\top \succeq  \Par{1 - 2C_4 \eps\log\eps^{-1}}\covar\\
\implies u_T^\top \mm(w^{(T)}) u_T \ge \Par{1 - 2C_4 \eps\log\eps^{-1}}\norm{\covar}_\infty.
\end{align*}
We used the Courant-Fischer characterization of eigenvalues, and that $u_T$ is a top eigenvector of $\mm(w^{(T)})$. Moreover, by termination conditions and Corollary~\ref{corr:oned-independence} (correctness of $\OneD$),
\[(1 + C\eps\log\eps^{-1}) u_T^\top \covar u_T \ge \sigma_T^2 \geq (1 + 5C_5\eps\log\eps^{-1})^{-1}u_T^\top \mm(w^{(T)}) u_T.\]
Combining these two bounds and rescaling yields the conclusion.
\end{proof}
Finally, we prove our main guarantee about Algorithm~\ref{alg:filter}.
\restatepolyfinal*
\begin{proof}
First, we will operate under Assumption~\ref{assume:total-filter}, which holds with probability at least $1 - \delta$. It is clear that the analyses of Lemma~\ref{lem:inv-holds} and~\ref{lem:output-correct} hold with $1 - \Theta(\eps\log\eps^{-1})$ multiplicative approximations of top eigenvector computation, which the power method approximates with high probability. Thus, each iteration takes time $O\Par{\frac{nd}{\eps} \log\frac{n}{\delta\eps}}$,
where we will union bound over the number of iterations. We now give an iteration bound: in any iteration where we do not terminate, Lemma~\ref{lem:uni-filter} implies
	\begin{align*}
	\sum_{i = 1}^n w^{(t-1)}_i - w^{(t)}_i &\geq \frac{1}{2\max_{i \in [n]} \inprod{u_t}{X_i}^2} \sum_{i = \ell}^n w_i a_i \\
	& \geq \frac{1}{2C_4 \log \frac{n}{\delta} \cdot \Tr (\covar)} \sum_{i = \ell}^n w_i a_i \\
	&\geq \frac{1}{2C_4 \log \frac{n}{\delta} \cdot \Tr (\covar)} \Par{\frac{\sum_{i = \ell}^n w_i}{\sum_{i \in [n] }w_i}} \sum_{i \in [n]} w_i a_i \\
	&= \Omega \Par{\eps \cdot \frac{\norm{\covar}_\infty}{\log \frac{n}{\delta}\cdot \Tr (\covar)}} = \Omega\Par{\frac{\eps}{d\log\frac{n}{\delta}}}.
	\end{align*}
	Here, the second line used Assumption~\ref{assume:total-filter}, the third used that the $a_i$ are in sorted order, and the last used the definition of $\ell$ as well as the derivations of Lemma~\ref{lem:output-correct} (specifically, that $\mm(w)$ spectrally dominates $(1 - O(\eps\log\eps^{-1}))\covar$ for roughly uniform $w$). The conclusion follows since there can be at most $O(d\log\tfrac{n}{\delta})$ iterations, as the algorithm terminates when a $2\eps$ fraction of the mass is removed, giving the overall runtime claim.
\end{proof}	 	%

\section{Deferred proofs from Section~\ref{sec:packing}}
\label{app:pack}

\subsection{Proofs from Section~\ref{sec:lp-lp}}

Since our notion of approximation is multiplicative, we can assume without more than constant loss that $\ma$ has bounded entries. This observation is standard, and formalized in the following lemma.

\begin{restatable}[Entrywise bounds on $\ma$]{lemma}{restateassumeabounded}\label{lem:assumeabounded}
	Feasibility of Problem~\ref{problem:lp-packing} is unaffected (up to constants in $\eps$) by removing columns of $\ma$ with entries larger than $n\eps^{-1}$.
\end{restatable}
\begin{proof}
	If $\ma_{ji} > n\eps^{-1}$ for any entry, then $x_i \le \tfrac{\eps(1 + \eps)}{n}$, else $\norm{\ma x}_p$ is already larger than $1 + \eps$. Ignoring all such entries of $x$ and rescaling can only change the objective by a $1 + O(\eps)$ factor.
\end{proof}

\restatelplppotential*
\begin{proof}
	Fix an iteration $t$. Define $\delta = \eta g_t$, and note $w_{t+1} = w_t + \delta \circ w_t$; henceforth in this proof, we will drop subscripts $t$ when clear. Observe that 
	\[
	\norm{\ma  w_{t+1}}_p = \norm{\ma ((1+\delta) \circ w)}_p = \left( \sum_{j \in [d]} [\ma w]_j^p \left(1 + \frac{[\ma (\delta \circ w_t)]_j}{[\ma w_t]_j}  \right)^p \right)^{1/p}.
	\]
	As $g \leq \1 \implies \delta \leq p^{-1}\1$, $\frac{\ma (\delta \circ w_t)}{\ma w_t} \leq p^{-1}$ entrywise. Via $(1+x)^p \leq \exp(px) \le 1 + px + p^2 x^2$ for $x \leq p^{-1}$, it follows that
	\[
	\norm{\ma ((1+\delta) \circ w)}_p \leq \left( \sum_{j \in [d]} [\ma w]_j^p \left(1 + \frac{p [\ma (\delta \circ w)]_j}{[\ma w]_j} + \left(\frac{p [\ma (\delta \circ w)]_j}{[\ma w]_j}\right)^2  \right) \right)^{1/p}.
	\]
	By direct manipulation of the above quantity, and recalling we defined $v = \tfrac{\ma w}{\norm{\ma w}_p}$,
	\begin{align*}
	\left(\sum_{j \in [d]}\left([\ma w]_j^p + p[\ma w]_j^{p - 1}[\ma (\delta \circ w)]_j + p^2[\ma w]_j^{p - 2}[\ma (\delta \circ w)]_j^2\right) \right)^{1/p}\\
	= \left(\norm{\ma w}_p^p \sum_{j \in [d]}\left(v_j^p + p v_j^{p-1} \frac{[\ma (\delta \circ w)]_j}{\norm{\ma w}_p} + p^2 v_j^{p-2} \left( \frac{[\ma (\delta \circ w)]_j}{\norm{\ma w}_p} \right)^2 \right)\right)^{1/p} \\
	=\norm{\ma w}_p \left(1+ \sum_{j \in [d]} \Par{p v_j^{p-1} \frac{[\ma (\delta \circ w)]_j}{\norm{\ma w}_p} + p^2 v_j^{p-2} \left( \frac{[\ma (\delta \circ w)]_j}{\norm{\ma w}_p}} \right)^2 \right)^{1/p}.
	\end{align*}
	Using $(1+x)^p > 1 + px$, i.e.\ $(1 + px)^{1/p} < 1 + x$, we thus obtain 
	\[
	\norm{\ma ((1+\delta) \circ w)}_p \leq \norm{\ma w}_p  + \inprod{v^{p-1}}{\ma (\delta \circ w)} + p \inprod{v^{p-1}}{\frac{(\ma (\delta \circ w))^2}{\ma w}}.
	\]
	Cauchy-Schwarz yields that $[\ma (\delta \circ w)]_j^2 \leq [\ma (\delta^2 \circ w)]_j [\ma w]_j$, $\forall j \in [d]$. Substituting into the above,
	\begin{equation}\label{eq:boundlpdiff}
	\begin{aligned}
	\norm{\ma ((1+\delta) \circ w)}_p &\leq \norm{\ma w}_p  + \inprod{v^{p-1}}{\ma (\delta \circ w)} + p \inprod{v^{p-1}}{\ma (\delta^2 \circ w)} \\
	&= \norm{\ma w}_p + \sum_{j \in [d]} \left[\ma^\top v^{p-1} \right]_j \delta_j w_j (1+ p \delta_j).
	\end{aligned}
	\end{equation}
	Finally, to bound this latter quantity, since $\delta = \eta g$, we observe that for all $j$ either $\delta_j = 0$ or $1 + p \delta_j = 1+ g_j = 2 - [\ma^\top v^{p-1}]_j$, in which case
	\[\left[\ma^\top v^{p-1} \right]_j(1 + p\delta_j) = \left[\ma^\top v^{p-1} \right]_j\left(2 - \left[\ma^\top v^{p-1} \right]_j\right) \le 1.\] 
	Thus, plugging this bound into \eqref{eq:boundlpdiff} entrywise, 
	\[\norm{\ma((1 + \delta)\circ w)}_p - \norm{\ma w}_p \le \sum_{j \in [d]} \delta_j w_j \left[\ma^\top v^{p-1} \right]_j(1 + p\delta_j) \le \sum_{j \in [d]} \delta_j w_j = \norm{w_{t + 1}}_1 - \norm{w_t}_1.\]
	Rearranging yields the desired claim.
\end{proof}

\subsection{Proofs from Section~\ref{sec:lp-sdp}}

Our analysis of Algorithm~\ref{alg:lp-sdp} will use the following helper fact.

\begin{lemma}[Spectral bounds on $\{\ma_i\}_{i \in [n]}$]\label{lem:assumeaspectralbound} Feasibility of Problem~\ref{problem:sdp-packing} is unaffected (up to constants in $\eps$) by removing matrices $\ma_i$ with an eigenvalue larger than $n\eps^{-1}$.
\end{lemma}
\begin{proof}
The proof is identical to Lemma~\ref{lem:assumeabounded}; we also require the additional fact that the Schatten norm $\norm{\cdot}_p$ is monotone in the Loewner order, forcing the constraint $x_i \le \tfrac{\eps(1 + \eps)}{n}$.
\end{proof}

We remark that we can perform this preprocessing procedure via power iteration on each $\ma_i$.

\restatesdppot*
\begin{proof}
	Drop $t$ and define $\delta = \eta g$. For simplicity, define the matrices
	\[\mm_0 \defeq \sum_{i \in [n]} w_i \ma_i,\; \mm_1 \defeq \sum_{i \in [n]} \delta_i w_i \ma_i,\; \mm_2 \defeq \sum_{i \in [n]} \delta_i^2 w_i \ma_i.\]
	We recall the Lieb-Thirring inequality $\Tr((\ma \mb\ma)^p) \le \Tr(\ma^{2p}\mb^p)$. Applying this, we have
	\[\norm{\mm_0 + \mm_1}_p^p = \Tr\Par{\left(\mm_0 + \mm_1\right)^p} \le \Tr\Par{\mm_0^p\left(\id + \mm_0^{-\half}\mm_1\mm_0^{-\half}\right)^p}.\]
	As $g \leq \1$, we have $\mm_0^{-\half}\mm_1 \mm_0^{-\half}\preceq p^{-1} \id$. Applying the bounds $(\id + \mm)^p \preceq \exp(p\mm) \preceq \id + p\mm + p^2 \mm^2$ for $\mm = \mm_0^{-\half}\mm_1 \mm_0^{-\half}$, where we use that $\id$ commutes with all $\mm$, it follows that
	\[\norm{\mm_0 + \mm_1}_p^p \le \Tr\Par{\mm_0^p + p\mm_0^{p - 1} \mm_1 + p^2 \mm_0^{p - 1} \mm_1 \mm_0^{-1} \mm_1}. \]
	Definitions of $\mm_0$, $\mm_1$, $\mm_2$, and preservation of positiveness under Schur complements imply
	\[\begin{pmatrix}\mm_0 & \mm_1 \\ \mm_1 & \mm_2 \end{pmatrix} \succeq 0 \implies \mm_2 - \mm_1 \mm_0^{-1} \mm_1 \succeq 0.\]
	Thus, $\mm_1\mm_0^{-1}\mm_1 \preceq \mm_2$. Applying this and recalling $\mv = \tfrac{\mm_0}{\norm{\mm_0}_p}$,
	\begin{align*}\norm{\mm_0 + \mm_1}_p^p &\le \Tr\Par{\mm_0^p + p\mm_0^{p - 1} \mm_1 + p^2 \mm_0^{p - 1} \mm_2}\\
	&=\norm{\mm_0}_p^p\left(1 + p\inprod{\mv^{p - 1}}{\frac{\mm_1}{\norm{\mm_0}_p} + \frac{p\mm_2}{\norm{\mm_0}_p}}\right). \end{align*}
	By $(1 + px)^{1/p} < 1 + x$, taking $p^{th}$ roots we thus have
	\[\norm{\mm_0 + \mm_1}_p \le \norm{\mm_0}_p + \inprod{\mv^{p - 1}}{\mm_1 + p\mm_2}. \]
	Finally, the conclusion follows as in Lemma~\ref{lemma:lp-potential}; by linearity of trace and $g = p\delta$,
	\[\inprod{\mv^{p - 1}}{\mm_1 + p\mm_2} = \sum_{i \in [n]} \inprod{\mv^{p - 1}}{\ma_i} \delta_i w_i(1 + p\delta_i) \le \sum_{i \in [n]}\delta_i w_i.\]
	Here, we used the inequality for all nonzero $g_i$,
	\[\inprod{\mv^{p - 1}}{\ma_i}(1 + p\delta_i) = \inprod{\mv^{p - 1}}{\ma_i}\left(2 - \inprod{\mv^{p - 1}}{\ma_i}\right) \le 1. \]
\end{proof}

\restatewgrowthsdp*
\begin{proof}
The proof is analogous to that of Theorem~\ref{thm:wgrowth}; we sketch the main differences here. By applying Lemma~\ref{lem:assumeaspectralbound} and monotonicity of Schatten norms in the Loewner order, we again have $\Phi_0 \le 1$, implying correctness whenever the algorithm terminates on Line 4. Correctness of dual certification again follows from lack of termination and the choice of $T$, as well as setting $u$ to indicate each coordinate. Finally, the returned matrix in Line 8 is correct by convexity of the Schatten-$q$ norm, and the fact that all $\mv_t^{p - 1}$ have unit Schatten-$q$ norm.

We now discuss issues regearding computing $g_t$ in Line 5 of the algorithm, the bottleneck step; these techniques are standard in the approximate SDP literature, and we defer a more formal discussion to e.g.\ \cite{JambulapatiLLPT20}. First, note that each coordinate of $g_t$ requires us to compute
\begin{equation}\label{eq:gradsdp}\frac{1}{\norm{\sum_{i \in [n]} [w_t]_i \ma_i}_p^{p - 1}} \cdot \inprod{\ma_i}{\Par{\sum_{i \in [n]} [w_t]_i \ma_i}^{p - 1}}.\end{equation}
We estimate the two quantities in the above expression each to $1 + \eps$ multiplicative error with high probability. Union bounding over iterations, and modifying Lemma~\ref{lem:sdp-potential} to use the potential $\norm{\sum_{i \in [n]} [w_t]_i \ma_i}_p - (1 + O(\eps))\norm{w_t}_1$, the analysis remains valid up to constants in $\eps$ with this multiplicative approximation quality. We now discuss our approximation strategies.

For shorthand, denote $\mm = \sum_{i \in [n]} [w_t]_i \ma_i$. To estimate the denominator of \eqref{eq:gradsdp}, it suffices to multiplicatively approximate $\norm{\mm}_p^p = \Tr[\mm^p]$ within a $1 + \eps$ factor, as raising to the $\tfrac{p-1}{p}$ power can only improve this. To do so, we use the well-known fact (e.g.\ \cite{DasguptaG03}) that letting $\mq$ be a $k \times d$ matrix with independent entries $\sim \Nor(0, \tfrac{1}{k})$, for $k = O(\tfrac{\log(\frac{nd}{\eps})}{\eps^2})$, with probability $1 - \textup{poly}((\tfrac{nd}{\eps})^{-1})$,
\[\Tr[\mm^p] \approx \sum_{\ell \in [k]} \mq_{\ell:} ^\top \mm^p \mq_{\ell:}\]
to a $1 + \eps$ factor. To read this from the standard Johnson-Lindestrauss guarantee, it suffices to factorize $\mm^p$ and use that each row of the square root's $\ell_2$ norm is preserved with high probability under multiplication by $\mq$, and then apply the cyclic definition of trace. Similarly, for each $i \in [n]$, we can approximate the numerators via
\[\Tr\Par{\mq \mm^{\frac{p - 1}{2}} \ma_i  \mm^{\frac{p - 1}{2}} \mq^\top}.\]
We can simultaneously compute all such quantities by first applying $O(p)$ matrix-vector multiplications through $\mm$ to each row of $\mq$, and then computing all quadratic forms. In total, the computational cost per iteration of all approximations is $O(\nnz \cdot \tfrac{p\log(\frac{nd}{\eps})}{\eps^2})$ as desired.
\end{proof}

\subsection{Proof of Proposition~\ref{prop:boxconstrainedp}}

In this section, following our prior developments, we prove the following claim.

\restateboxconstrainedp*

\subsubsection{Reduction to a decision problem}\label{sssec:reducedecision}

Given access to an oracle for the following approximate decision problem, we can implement an efficient binary search for estimating $\opt$. Specifically, letting the range of $\opt$ be $(\mu_{\text{lower}}, \mu_{\text{upper}})$, we can subdivide the range into $O(\tfrac{1}{\eps}\log\tfrac{\mu_{\text{upper}}}{\mu_{\text{lower}}})$ multiplicative intervals of range $1 + \eps$, and then compute a binary search using our decision oracle. This incurs a multiplicative $\log(\tfrac{nd}{\eps})$ overhead in the setting of Proposition~\ref{prop:boxconstrainedp} (see Appendix A, \cite{JambulapatiLLPT20}, for a more formal treatment).

\begin{problem}\label{problem:boxconstrainedp}
Given $\{\ma_i\}_{i \in [n]} \in \PSD^d$, either find primal solution $x \in \Delta^n$ with $\norm{\alla(x)}_{p} \le 1 + \eps$, $\norm{x}_\infty \le \tfrac{(1 + \eps)(1 + \alpha)}{n}$, or conclude no $x \in \Delta^n$ satisfies $\norm{\alla(x)}_{p} \le 1 - \eps$, $\norm{x}_\infty \le \tfrac{(1 - \eps)(1 + \alpha)}{n}$.
\end{problem}

The hard constraint $\norm{x}_\infty \le \tfrac{1 + \alpha}{n}$ in the definition \eqref{eq:boxconstrainedschatten} can be adjusted by constant factors to admit the $\ell_\infty$ bound in Problem~\ref{problem:boxconstrainedp}, since we assumed $\eps = O(\alpha)$ is sufficiently small.

\subsubsection{Preliminaries}

We use the shorthand $\ms \defeq \tfrac{n}{1 + \alpha} \id$, and $p' \defeq \tfrac{\log n}{\eps}$, so $\ell_{p'}$ and $\ell_\infty$ are interchangeable up to $1 + O(\eps)$ factors. In other words, Problem~\ref{problem:boxconstrainedp} asks to certify whether there exists $x \in \Delta^n$ with
\begin{equation}\label{eq:pinfdecision}\max\Par{\norm{\alla(x)}_p,\;\norm{\ms x}_{p'}} \le 1,\end{equation}
up to multiplicative $1 + \eps$ tolerance on either side. Consider the potential function
\begin{equation}\label{eq:phidef}\Phi(w) \defeq \log\Par{\exp\Par{\norm{\alla(w)}_p} + \exp\Par{\norm{\ms w}_{p'}}}- \norm{w}_1.\end{equation}
It is clear that the first term of $\Phi(w)$ approximates the left hand side of \eqref{eq:pinfdecision} up to a $\log 2$ additive factor, so if any of $\norm{\alla(w)}_p$, $\norm{\alla(w)}_{p'}$, or $\norm{w}_1$ reaches the scale $3\eps^{-1}$ and $\Phi(w)$ is bounded by $1$, we can safely terminate. and conclude primal feasibility for Problem~\ref{problem:boxconstrainedp}. Next, we compute
\begin{equation}\label{eq:gradphi}
\begin{aligned}
\nabla_i \Phi(w) = 1 - \frac{\exp\Par{\norm{\alla(w)}_{p}}\inprod{\ma_i}{\my(w)} + \exp\Par{\norm{\ms w}_{p'}}[\ms z(w)]_i}{\exp\Par{\norm{\alla(w)}_{p}} + \exp\Par{\norm{\ms w}_{p'}}} \text{ for all } i \in [n],\\
\text{where } \my(w) \defeq \Par{\frac{\alla(w)}{\norm{\alla(w)}_p}}^{p - 1},\; z(w) \defeq \Par{\frac{\ms w}{\norm{\ms w}_{p'}}}^{p' - 1}
\end{aligned}
\end{equation}
The following helper lemma will be useful in concluding dual infeasibility of Problem~\ref{problem:boxconstrainedp}.

\begin{lemma}\label{lem:inprodatleasteps}
In the setting of Problem~\ref{problem:boxconstrainedp}, suppose there exists $x^* \in \Delta^n$ with
\[\norm{\alla(x^*)}_p \le 1 - \eps,\; \norm{\ms x^*}_{p'} \le 1 - \eps.\]
Then, for any $w$,
\[\inprod{\nabla \Phi(w)}{x^*} \ge \eps.\]
\end{lemma}
\begin{proof}
From the definitions in \eqref{eq:gradphi}, it is clear that $\norm{\my(w)}_q = \norm{z(w)}_{q'} = 1$, where $q$, $q'$ are the dual norms of $p$, $p'$ respectively. Moreover, by the definition of $x^*$, we have for all $\norm{\my}_q = \norm{z}_{q'} = 1$,
\[\inprod{\my}{\alla(x)} \le 1 - \eps,\; \inprod{z}{\ms x} \le 1 - \eps.\]
This follows from the dual definition of the $\ell_p$ norm (see Fact~\ref{fact:qnormcert}). Now, note that for some nonnegative $\alpha(w)$, $\beta(w)$ summing to 1, using the above claim and \eqref{eq:gradphi},
\[\inprod{\nabla \Phi(w)}{x^*} = 1 - \Par{\alpha(w) \inprod{\my(w)}{\alla(x^*)} + \beta(w) \inprod{z(w)}{\ms x^*}} \ge \eps,\]
as desired (here, we used positivity of all relevant quantities).
\end{proof}

\subsubsection{Potential monotonicity}

We prove a monotonicity property regarding the potential $\Phi$ in \eqref{eq:phidef}.

\begin{lemma}
Let $w \in \R^n_{\ge 0}$ satisfy $\norm{\alla(w)}_p \le 3\eps^{-1}$, $\norm{\ms w}_{p'} \le 3\eps^{-1}$, let $g = \max(0, \nabla \Phi(w))$ entrywise, and let $w' = (1 + \eta g) \circ w$, where $\eta = (4p')^{-1}$. Then, $\Phi(w') \le \Phi(w)$.
\end{lemma}
\begin{proof}
Denote for simplicity the threshold $K = 3\eps^{-1}$ and the step vector $\delta = \eta g$. First, by prior calculations in Lemma~\ref{lemma:lp-potential} and Lemma~\ref{lem:sdp-potential}, it follows that
\begin{align*}
\norm{\alla(w')}_p \le \norm{\alla(w)}_p + \Delta_{\alla},\; \norm{\ms w'}_{p'} \le \norm{\ms w}_{p'} + \Delta_{\ms},\\
\text{where } \Delta_{\alla} \defeq \sum_{i \in [n]} \inprod{\ma_i}{\my(w)} \delta_i w_i (1 + p\delta_i),\; \Delta_{\ms} \defeq \sum_{i \in [n]} [\ms z(w)]_i \delta_i w_i (1 + p'\delta_i).
\end{align*}
Next, note that by $\delta \le \eta$ entrywise and lack of termination (i.e.\ the threshold $K$),
\begin{align*}
\Delta_{\alla} \le (1 + p\eta) \eta \inprod{\my(w)}{\alla(w)} \le 2\eta\norm{\alla(w)}_{p} \le 1.
\end{align*}
Therefore, by $\exp(x) \le 1 + x + x^2$ for $x \le 1$, 
\begin{equation}\label{eq:expagrowth}\exp\Par{\norm{\alla(w')}_p} \le \exp\Par{\norm{\alla(w)}_p}\Par{1 + \Delta_{\alla} + \Delta_{\alla}^2}.\end{equation}
Moreover, by applying Cauchy-Schwarz and the threshold $\norm{\alla(w)}_p \le K$ once more,
\begin{equation}\label{eq:secondorderbounda}
\begin{aligned}
\Delta_{\alla}^2 &\le (1 + p\eta)^2\Par{\sum_{i \in [n]} \inprod{\ma_i}{\my(w)} \delta_i w_i}^2 \\
&\le 2\Par{\sum_{i \in [n]} \inprod{\ma_i}{\my(w)} \delta_i^2 w_i} \inprod{\my(w)}{\alla(w)} \le 2K\Par{\sum_{i \in [n]} \inprod{\ma_i}{\my(w)} \delta_i^2 w_i}.
\end{aligned}
\end{equation}
Combining \eqref{eq:expagrowth} and \eqref{eq:secondorderbounda} (and applying similar reasoning to the term $\Delta_{\ms}$), we conclude
\begin{align*}
\exp\Par{\norm{\alla(w')}_p} \le \exp\Par{\norm{\alla(w)}_p} \Par{1 + \sum_{i \in [n]} \inprod{\ma_i}{\my(w)} \delta_i w_i (1 + (p + 2K)\delta_i)}, \\
\exp\Par{\norm{\ms w'}_{p'}} \le \exp\Par{\norm{\ms w}_{p'}} \Par{1 + \sum_{i \in [n]} [\ms z(w)]_i \delta_i w_i (1 + (p' + 2K)\delta_i)}.
\end{align*}
Recall the inequality $\log(1 + x) \le x$ for nonnegative $x$. Expanding the definition of $\Phi$ and $\nabla \Phi$ (cf. \eqref{eq:phidef}), and plugging in the above bounds, we conclude that 
\begin{align*}
\Phi(w') - \Phi(w) &= \log\Par{\frac{\exp\Par{\norm{\alla(w')}_p} + \exp\Par{\norm{\ms w'}_{p'}}}{\exp\Par{\norm{\alla(w)}_p} + \exp\Par{\norm{\ms w}_{p'}}}} - \inprod{\delta}{w} \\
&\le \sum_{i \in [n]} (1 - \nabla_i \Phi(w))\delta_i w_i(1 + (p' + 2K)\delta_i) - \inprod{\delta}{w} \\
&= \sum_{i \in [n]} \Par{(1 - \nabla_i \Phi(w)) (1 + (p' + 2K)\delta_i) - 1}\delta_i w_i.
\end{align*}
As before, we show that this sum is entrywise nonpositive. For any $i \in [n]$ with $\delta_i \neq 0$, we have
\begin{align*}(1 - \nabla_i \Phi(w)) (1 + (p' + 2K)\delta_i) - 1 &= (1 - \nabla_i \Phi(w)) (1 + (p' + 2K)\eta\nabla_i \Phi(w)) - 1 \\
&\le  (1 - \nabla_i \Phi(w))(1 + \nabla_i \Phi(w)) - 1 \le 0,\end{align*}
as desired, where we used that $\eta^{-1} \ge p' + 2K$. This yields the conclusion $\Phi(w') \le \Phi(w)$.
\end{proof}
\subsubsection{Algorithm and analysis}

Finally, we state Algorithm~\ref{alg:boxpack} and prove Proposition~\ref{prop:boxconstrainedp}.

\begin{algorithm}
	\caption{$\BoxPacking(\{\ma_i\}_{i \in [n]}, \eps, p, \alpha)$}
	\begin{algorithmic}[1]\label{alg:boxpack}
		\STATE \textbf{Input:} $\{\ma_i\}_{i \in [n]} \in \PSD^d, \eps \in [0, \thalf], p \ge 2$, $\alpha \in [0, n - 1]$
		\STATE $p' \gets \tfrac{\log n}{\eps}$, $\ms \gets \tfrac{n}{1 + \alpha}\id$
		\STATE $\eta \gets (4p')^{-1}$, $K \gets 3\eps^{-1}$, $T \gets \frac{6\log(\frac{nd}{\eps})}{\eta\eps}$
		\STATE $[w_0]_i \gets \tfrac{\eps}{n^2 d}$ for all $i \in [n]$, $t \gets 0$
		\WHILE{$\norm{\alla(w_t)}_p, \norm{\ms w_t}_{p'}, \norm{w_t}_1 \le K$}	
		\STATE $g_t \gets \max\Par{0, \nabla \Phi(w_t)}$ entrywise, where we use the definition \eqref{eq:phidef}
		\STATE $w_{t + 1} \gets w_t \circ (1 + \eta g_t)$, $t \gets t+1$
		\IF{$t \geq T$}
		\RETURN{Infeasible}
		\ENDIF
		\ENDWHILE
		\RETURN{ $x = \frac{w_t}{\norm{w_t}_1}$ }
	\end{algorithmic}
\end{algorithm}

\begin{proof}[Proof of Proposition~\ref{prop:boxconstrainedp}]
Correctness of the reduction to deciding Problem~\ref{problem:boxconstrainedp} follows from the discussion in Section~\ref{sssec:reducedecision}. Moreover, by the given Algorithm~\ref{alg:boxpack}, it is clear (following e.g.\ the preprocessing of Lemma~\ref{lem:assumeaspectralbound}) that $\Phi(w_t) \le 1$ throughout the algorithm, so whenever the algorithm terminates we have primal feasibility. It suffices to prove that whenever the problem admits $x^*$ with
\[\norm{\alla(x^*)}_p \le 1 - \eps,\; \norm{\ms x^*}_{p'} \le 1 - \eps,\]
then the algorithm terminates on Line 5 in $T$ iterations. Analogously to Theorem~\ref{thm:wgrowth}, we have
\[\eta(1 - \eta) \sum_{0 \le t < T} \inprod{g_t}{x^*} \le \log n - \log \norm{w_0}_1 + \log \norm{w_T}_1 \le 2\log\Par{\frac{nd}{\eps}} + \log \norm{w_T}_1.\]
Next, since $g_t$ is an upwards truncation of $\nabla \Phi(w_t)$, applying Lemma~\ref{lem:inprodatleasteps} implies that
\[\norm{w_T}_1 \ge \exp\Par{\frac{\eta\eps T}{2} - 2\log\Par{\frac{nd}{\eps}}}.\]
The conclusion follows by the definition of $T$, as desired. Finally, the iteration complexity follows analogously to the discussion in Theorem~\ref{thm:wgrowthsdp}'s proof, where the only expensive cost is estimating coordinates of the $\alla$ component of $\nabla \Phi(w_t)$ every iteration.
\end{proof}

Finally, we remark that by opening up the dual certificates $\my(w)$, $\mz(w)$ of our mirror descent analysis, we can in fact implement a stronger version of the decision Problem~\ref{problem:boxconstrainedp} which returns a feasible dual certificate whenever the primal problem is infeasible. We omit this extension for brevity, as it is unnecessary for our applications, but it is analogous to the analysis of Theorem~\ref{thm:wgrowthsdp}.

\section{Deferred proofs from Section~\ref{sec:subspace}}
\label{app:pca}

\subsection{Proof of Proposition~\ref{prop:poweriter}}

\poweriter*

\begin{proof}
	We claim that Algorithm 1 in~\cite{MuscoM15} applied to the matrix $\ma^p$ with a careful choice of exponent $q$ in their Algorithm 1 yields this guarantee.
	Specifically, we choose $q_1, q_2$, both of which satisfy the criteria in their main theorem, such that the iterates produced by simultaneous power iteration $\mm^p$ with exponent $q_1$ and $\mm^{p - 1}$ with exponent $q_2$ are identical; it suffices to choose $q$ a multiple of $p(p - 1)$.
	Thus, we can also apply their guarantees to $\ma^{p - 1}$ and apply a union bound. Notice that their Algorithm 1 also contains some postprocessing to ensure that they obtain singular values in the right space, which is unnecessary for us, as our matrices are Hermitian.
\end{proof}

\subsection{Proof of Lemma~\ref{lem:robust-objective-ub}}
\robustobjectiveub*

\begin{proof}
Lemma~\ref{lem:pnormbound_fs} implies that letting $w^*$ be the uniform distribution over the uncorrupted samples amongst $X_1, \ldots, X_n$, we have with probability at least $1 - \tfrac{\delta}{2}$, and denoting $\teps \defeq 2C_3 \cdot \eps\log\eps^{-1}$,
	\[\norm{\sum_{i \in [n]} w^*_i X_i X_i^\top}_p \le \Par{1 + \frac{\teps}{2}}\norm{\covar}_p.\]
	Therefore, the mixed $\ell_\infty$-$\ell_p$ packing semidefinite program
	\[\exists w^* \in \Delta^n \text{ with } \norm{w^*}_\infty \le \frac{1}{(1 - \eps)n},\; \norm{\sum_{i \in [n]} w^*_i X_i X_i^\top}_p \le \Par{1 + \frac{\teps}{2}}\norm{\covar}_p\]
	is feasible. This completes the proof.
\end{proof}

\subsection{Proof of Lemma~\ref{lem:linfguarantee}}

\linfguarantee*

\begin{proof}
We follow the notation of \eqref{eq:mgbdef}. First, by the guarantees of Corollary~\ref{cor:robust-output-ub},
	\[\wg = 1 - \wb \ge 1 - \frac{\eps n}{(1 - 2\eps) n} = \frac{1 - 3\eps}{1 - 2\eps} \ge 1 - 2\eps.\] 
Therefore, again applying Corollary~\ref{cor:robust-output-ub}, for all $i \in G$,
	\[\frac{w_i}{\wg} \le \frac{1}{(1 - 2\eps) n} \cdot \frac{1 - 2\eps}{1 - 3\eps} = \frac{1}{(1 - 3\eps) n}.\]
We conclude that the set of weights $\{\tfrac{w_i}{\wg}\}_{i \in G}$ belong to  $\fS_{3\eps}^{(1 - \eps)n}$. By applying Corollary~\ref{corr:infnormbound_fs} to these weights and adjusting the definition of $C_3$ by a constant, we conclude with probability at least $1 - \tfrac{\delta}{2}$
	\[\Par{1 + C_3 \cdot \eps\log\eps^{-1}}\covar \succeq \sum_{i \in G} \frac{w_i}{\wg} X_i X_i^\top \succeq  \Par{1 -  C_3 \cdot \eps\log\eps^{-1}}\covar.\]
	The conclusion follows by multiplying through by $w_G$, and using the definition $\teps = 2C_3 \cdot \eps\log\eps^{-1}$.
\end{proof} 	\end{appendix}

\end{document}